\documentclass[manuscript,nonacm]{acmart}

\usepackage{balance}
\usepackage{enumitem}
\usepackage{listings}
\usepackage{svg}
\usepackage{booktabs}
\usepackage{subcaption}

\setcopyright{none}
\renewcommand\footnotetextcopyrightpermission[1]{}
\AtBeginDocument{%
  }

\newcommand{\myparagraph}[1]{\medskip{\noindent\bfseries\itshape{#1}.~}}

\begin{document}

\title{Parallel batch queries on dynamic trees: algorithms and experiments}

\author{Humza Ikram}
\email{humzai@andrew.cmu.edu}
\affiliation{%
  \institution{Carnegie Mellon University}
  \city{Pittsburgh}
  \country{USA}
}

\author{Andrew Brady}
\email{acbrady@andrew.cmu.edu}
\affiliation{
  \institution{Carnegie Mellon University}
  \city{Pittsburgh}
  \country{USA}
}

\author{Daniel Anderson}
\email{dlanders@andrew.cmu.edu}
\affiliation{%
  \institution{Carnegie Mellon University}
  \city{Pittsburgh}
  \country{USA}
}

\author{Guy Blelloch}
\email{guyb@andrew.cmu.edu}
\affiliation{%
  \institution{Carnegie Mellon University}
  \city{Pittsburgh}
  \country{USA}
}

\renewcommand{\shortauthors}{Humza et al.}

\begin{abstract}
  Dynamic tree data structures maintain a forest while supporting
  insertion and deletion of edges and a broad set of queries in
  $O(\log n)$ time per operation.  Such data structures are at the
  core of many modern algorithms.  Recent work has extended dynamic
  trees so as to support batches of updates or queries so as to run in
  parallel, and these batch parallel dynamic trees are now used in
  several parallel algorithms.

  In this work we describe improvements to batch parallel dynamic
  trees, describe an implementation that incorporates these
  improvements, and experiments using it.  The improvements includes
  generalizing prior work on RC (rake compress) trees to work with
  arbitrary degree while still supporting a rich set of queries, and
  describing how to support batch subtree queries, path queries,
  LCA queries, and nearest-marked-vertex queries in
  $O(k + k \log (1 + n/k))$ work and polylogarithmic span.  Our implementation
  is the first general implementation of batch dynamic trees
  (supporting arbitrary degree and general queries).

  Our experiments include measuring the time to create the trees,
  varying batch sizes for updates and queries, and using the tree
  to implement incremental batch-parallel minimum spanning trees.  To
  run the experiments we develop a forest generator that is
  parameterized to create distributions of trees of differing
  characteristics (e.g., degree, depth, and relative tree sizes).  Our
  experiments show good speedup and that the algorithm performance is
  robust across forest characteristics.
\end{abstract}



\maketitle

\section{Introduction}

The dynamic trees problem is to maintain a forest of trees while
supporting insertion and deletions of edges, and a variety of queries
on the tree.  Queries can include, for example, determining whether two vertices are in
the same tree, the minimum weight on a path in a tree, the sum of
values on a subtree, and many more.  The problem was introduced by
Sleator and Tarjan~\cite{sleator1983data}, where they presented a data structures
that supports updates as well as path queries (e.g., max weight or sum
of weights on a path) in $O(\log n)$ amortized time per operation on a tree with $n$ vertices.
They also present important applications of dynamic trees, including to
incremental minimum spanning trees, maximum flow, and network simplex.
In all cases the dynamic trees improved the bounds over prior results.

Since the work of Sleator and Tarjan, there have been many improvements
and
extensions~\cite{henzinger1995randomized,frederickson1985data,frederickson1997ambivalent,frederickson1997data,holm1998top,tarjan2005self,alstrup2005maintaining,acar2004dynamizing,acar2005experimental}.
Furthermore, researchers have identified many more applications of
dynamic trees, including algorithms for dynamic graph
connectivity~\cite{frederickson1985data,henzinger1995randomized,holm2001poly,kapron2013dynamic,acar2019parallel},
dynamic minimum spanning
trees~\cite{frederickson1985data,henzinger1995randomized,henzinger2001maintaining,holm2001poly,anderson2020work},
and minimum
cuts~\cite{karger2000minimum,geissmann2018parallel,gawrychowski2019minimum,anderson2023cuts}.
Importantly, dynamic trees play a key role in recent groundbreaking
results on maximum flow in almost linear time~\cite{maxflow23}.

Most of this work has considered sequential algorithms that permit a single update each time step.
Given that all modern machines are now
parallel, it is important to support processing multiple updates in
parallel.  
Thus, more recent work has sought to maintain a forest in the batch-parallel dynamic model, 
where edges are inserted and deleted in batches that must be handled in parallel.
 The goal of
such batch dynamic data structures is typically to ensure the work per
update is no more than in the sequential context (i.e., work
efficient) and that the span is polylogarithmic in the
size of the forest.   The span is the critical path of dependences and hence
the minimum time that is required to run a computation in parallel.

The first such work was Tseng et al.'s batch-parallel
Euler-tour trees~\cite{tseng2019batch}, which for batch updates of size
$k$, achieves $O(k \log(1 + \frac{n}{k}))$ expected work, which matches the
sequential algorithms ($O(\log n)$ work per edge) for low values of $k$,
and is optimal ($O(n)$ work) for large values of $k$.  Furthermore, the
span for processing each batch is $O(\log n)$.  Their work also
included an implementation that showed that maintaining Euler tour
trees is efficient in practice.  However, Euler-tour trees are quite
limited in the set of operations they support.  For example, they do
not support path queries, or subtree queries where the operator does
not have an inverse (e.g., find the maximum in a subtree).

Acar et al. developed parallel batch-dynamic
Rake-Compress-trees~\cite{acar2020batch} (RC-trees) with the same cost bounds.  RC-trees
support a much broader class of queries, including all queries that
are supported by any of the other known dynamic trees.  Batch parallel
RC-trees have proven to be a powerful technique in developing other
parallel algorithms, including batch parallel incremental
MST~\cite{anderson2020work}, parallel minimum
cut~\cite{anderson2023cuts}, and parallel breadth-first
search~\cite{ghaffari2023nearly}.    
The original work on RC-trees used randomization.   Recent work has
show that the same work bounds can be achieved deterministically~\cite{anderson2024}.
RC-trees, however, only directly
support trees with constant vertex degree.  Prior work has had to apply
ad-hoc techniques to simulate trees with higher degree.
Furthermore only a minimal implementation of RC-trees exists~\cite{anderson2021efficient}.

In this paper we present several results that build on RC-trees.
We show how to layer $d$-ary trees on top of trees with constant
degree while still supporting dynamic updates and queries. We maintain the same bounds as 
with constant degree, at the cost of requiring randomization. 
We also present some new better
bounds on queries.  In particular we describe algorithms for batch Lowest Common Ancestor (LCA)
queries, batch subtree queries, batch patch queries, and batch nearest-marked-vertex queries that run in $O(k + k \log (1+\frac{n}{k} ))$ work
where $k$ is the batch size and $n$ is the number of vertices in the
forest.  The span for these queries is $O(\log n)$.  This improves
over simply running the queries independently, which would require
$O(k \log n)$ work. Theoretically, we achieve the following result:

\begin{theorem}
There exists a data structure that support the following operations on
arbitrary weighted forests:
\begin{itemize}[leftmargin=*]
  \item A batch of $k$ edge insertions and deletions in $O(k+k \log (1 + \frac{n}{k}))$
    work and $O(\log^2 n)$ span.
   \item A batch of $k$ subtree queries over weights from a commutative semigroup in $O(k \log (1 + \frac{n}{k}))$
   work and $O(\log n)$ span.
   
   \item A batch of $k$ path queries over weights from a commutative group, or min/max queries in $O(k+k \log (1 + \frac{n}{k}))$ work and $O(\log n)$ span.
   
   \item A batch of $k$ LCA queries in  $O(k+k \log (1 + \frac{n}{k}))$
     work and $O(\log n)$ span.
   \item A batch of $k$ nearest-marked vertex queries in  $O(k \log (1 + \frac{n}{k}))$
     work and $O(\log n)$ span.
   \end{itemize}
   Here $n$ is the number of vertices in the forest. If the forest has bounded degree, these bounds are deterministic. If the forest has unbounded degree then our bounds are randomized and the work holds in expectation while the span bounds hold w.h.p.
\end{theorem}

\noindent In addition to theoretical contributions, we have implemented RC-trees along with several additional features. We have implemented the mapping from arbitrary degree to constant degree trees.
We have implemented batch subtree and LCA queries.
Furthermore,
we have implemented compressed path trees using RC-trees.   Given a subset
of the vertices, a compressed path tree generates a new tree that has
size proportional to the number of marked vertices and maintains path
lengths within each tree.   These are used both in batch incremental
MST~\cite{anderson2020work} and in parallel minimum
cut~\cite{anderson2023cuts}.

We run a suite of experiments on these trees. This includes
developing a forest-generator that allows us to generate forests with
a wide range of characteristics, including having widely varying tree
size, and having both deep and shallow trees. The experiments we run
include:
\begin{enumerate}[leftmargin=*]
\item Building the data structure given a set of edges.  
\item Batch inserting and deleting edges.  
\item Batch queries: path, subtree, batch subtree, and LCA.
\item Incremental MST.  We measure the performance of incremental
  MST, which uses compressed trees.   
  \end{enumerate}

\noindent Experiments 1, 2, and 4 are run with four very four different forest structures; Experiment 2 is run with 2 of these forest structures. In experiment 1, we vary $n$, whereas in experiments 2, 3, and 4, we fix $n$ and vary the batch size. In all experiments we vary the number of processors.
We plan to release this code for public use and believe individual components of our contribution (such as the general purpose ternarization algorithm and the tree generation scheme) have use cases beyond batch dynamic RC trees.

\section{Background}

\subsection{Analysis and Models of Computation}

\myparagraph{The PRAM and work-span analysis} 
We use the PRAM (Parallel Random Access Machine) model with work-span analysis~\cite{JaJa92,Blelloch96}.   In particular, we assume an algorithm has access to an unbounded shared memory and takes a sequence of $s$ steps.  Each step $i$ performs $w_i$ constant-time operations in parallel (each has access to its index, $[0,\ldots,w_i)$).   The steps are executed one after
the other sequentially.
The \emph{work} of the algorithm is the total number of operations performed, i.e., $\sum_{i=1}^s w_i$, and the \emph{span} is the number of steps $s$.    Any algorithm with $W$ work and $S$ span can be scheduled on
a traditional $p$-processor PRAM~\cite{pram78} in $O(W/p + S)$ time~\cite{brent1974parallel}, and any $p$ processor PRAM algorithm that takes $T$ time, does $O(pT)$ work and has $O(T)$ span. Variants of the PRAM differ in their assumptions
about whether or not concurrent operations within a single step can access the same memory location. \emph{Common}-\textsc{CRCW PRAM} permits concurrent writes but requires that all concurrent writes to the same location write the same value,
else the computation is invalid. The algorithms that we use as subroutines were designed for the Common-CRCW PRAM.

\myparagraph{Simulations} We note that the PRAM model can be mapped onto less synchronous models.  For example any PRAM algorithm can be mapped onto the binary-forking model while preserving the work and increasing the span by a factor of $O(\log n)$~\cite{blelloch2020optimal}.    It is likely the span for some of the algorithms in this paper can be improved for the binary-forking model over this naive simulation.

\myparagraph{Parallel primitives} Given a list of length $n$ and operator $f$ with an identity, a prefix sum accumulates the value of $f$ as $f$ is applied to the elements of the list, left to right. Prefix sums take $O(n)$ work and $O(\log n)$ span \cite{blelloch1990pre}. A filter operation, which returns a subset of a list that satisfies some predicate, can be implemented with prefix sums in $O(n)$ work and $O(\log n)$ depth \cite{blelloch1990pre}. Flatten, which takes a 2D list and concatenates each element to return a 1D list, uses prefix sums and takes $O(n)$ work and $O(\log n)$ depth \cite{blelloch1990pre}. 

Parallel semisort groups equal values together. However, a semisort does not sort the values from least to greatest. Parallel semisorting takes $O(n)$ expected work and $O(\log n)$ span w.h.p. \cite{Valiant91}. Using semisort, we implement a groupBy, which, given a sequence of pairs, groups the pairs with the same first element together.

We use a parallel hash table which can support parallel look-ups, parallel deletes and parallel inserts in $O(k)$ expected work and $O(\log^* k)$ span w.h.p. in $k$ \cite{Gil91a}. We also use parallel list contraction to remove marked vertices from a linked list, which takes $O(n)$ expected work and has $O(\log n)$ depth w.h.p. \cite{cole1986approximate}.

\subsection{Rake-Compress Trees}

\myparagraph{Parallel tree contraction}
Tree contraction is a procedure for computing functions over trees in parallel in low span~\cite{miller1985parallel}. It involves repeatedly applying \emph{rake} and \emph{compress} operations to the tree while aggregating data specific to the problem. The rake operation removes a leaf from the tree and aggregates its data with its parent. The compress operation replaces a vertex of degree two and its two adjacent edges with a single edge joining its neighbors, aggregating any data associated with the vertex and its two adjacent edges.

Rake and compress operations can be applied in parallel as long as they are applied to an independent set of vertices. Miller and Reif~\cite{miller1985parallel} describe a linear work and $O(\log n)$ span randomized algorithm for tree contraction. Their algorithm performs a set of \emph{rounds}, each round raking every leaf and an independent set of degree two vertices by flipping coins. They show that it takes $O(\log n)$ rounds to contract any tree to a singleton w.h.p. They also describe a deterministic algorithm, but it is not work efficient. Later, Gazit, Miller, and Teng~\cite{gazit1988optimal} obtain a work-efficient deterministic algorithm with $O(\log n)$ span.

These algorithms are defined for constant-degree trees, so non-constant-degree trees are handled by converting them into bounded-degree equivalents, e.g., by ternarization~\cite{frederickson1985data}.

\myparagraph{Rake-Compress Trees (RC-Trees)}
RC-Trees~\cite{acar2005experimental,acar2020batch} are based on viewing the process of parallel tree contraction as inducing a \emph{clustering}. A cluster is a connected subset of edges and vertices. The base clusters are the $n+m$ singletons containing the individual edges and vertices. The internal (non-base) clusters that arise have the following properties:

\begin{enumerate}[leftmargin=*]
	\item The subgraph induced by the vertex subset is connected,
	\item the edge subset contains all of the edges in the subgraph induced by the vertex subset,
	\item every edge in the edge subset has at least one endpoint in the vertex subset.
\end{enumerate}

\noindent This makes them somewhat of a hybrid of topology trees~\cite{frederickson1985data,frederickson1997ambivalent,frederickson1997data} and top trees~\cite{holm1998top,tarjan2005self,alstrup2005maintaining}, which cluster just the vertices or just the edges respectively.  Importantly but somewhat unintuitively, an RC cluster may contain an edge without containing both endpoints of that edge.  A vertex that is an endpoint of an edge, but is not contained in the same cluster as that edge is called a \emph{boundary vertex} of the cluster containing the edge.

The clusters of the RC-Tree always have at most two boundary vertices, and hence can be classified as \emph{unary clusters}, \emph{binary clusters}, or \emph{nullary clusters}. Unary clusters
arise from rake operations and have one boundary vertex. Binary clusters arise from compress operations and have two boundary vertices. A binary cluster with boundary vertices $u$ and $v$ always corresponds to an edge $(u,v)$ in the corresponding round of tree contraction.  Binary clusters can therefore be thought of as ``generalized edges'' (this notion is also used by top trees~\cite{alstrup2005maintaining}). 

To form a recursive clustering from a tree contraction, we begin with the base clusters and the uncontracted tree. On each round, for each vertex $v$ that contracts via rake or compress (which, remember, form an independent set), we identify the set of clusters that are adjacent to $v$ (equivalently, all clusters that have $v$ as a boundary vertex). These clusters are merged into a single cluster consisting of the union of their contents. We call $v$ the \emph{representative} vertex of the resulting cluster. The boundary vertices of the resulting cluster are the union of the boundary vertices of its constituents, minus $v$.

Rake operations always create unary clusters and compress operations always create binary clusters. When a vertex has no neighbors, it \emph{finalizes} and creates a nullary cluster representing the root of its connected component. Since each vertex rakes, compresses, or finalizes exactly once, there is a one-to-one mapping between representative vertices of the original tree and internal clusters. 

An RC-Tree then encodes this recursive clustering. The leaves of the RC-Tree are the base edge clusters. Since the base cluster for vertex $v$ is always a direct child of the cluster for which $v$ is the representative, we can omit base vertex clusters from the RC-tree without losing information. Internal nodes of the RC-Tree are clusters formed by tree contraction, such that the children of a node are the clusters that merged to form it. The root of the RC-Tree is a cluster representing the entire tree, or connected component in the case of a disconnected forest.

\myparagraph{Lowest Common Ancestor (LCA)} The LCA is relevant to batch query algorithms on RC-trees. \begin{definition} \label{ancDef} Given an unrooted tree $T$ and vertices $u$, $v$, and $r$, we say that a vertex $c$ is the LCA of $u$ and $v$ in $T$ with respect to the root $r$, denoted $LCA_T(u,v,r')$, if
\begin{enumerate}[leftmargin=*]
\item $c$ is an ancestor of $u,v$ 
\item For all $c'$ such that $c'$ is an ancestor of $u,v$, $c'$ is an ancestor of $c$
\end{enumerate} LCA can be generalized to a forest, where a null value is returned if $u,v,r$ are not in the same component. \end{definition}





\newcommand{\commonboundary}{common boundary}
\newcommand{\Commonboundary}{Common boundary}
\newcommand{\CommonBoundary}{Common Boundary}

\newcommand\textproc{\textsc}


\newcommand{\linkop}{\textbf{\textproc{Link}}}
\newcommand{\cutop}{\textbf{\textproc{Cut}}}
\newcommand{\connectedop}{\textbf{\textproc{connected}}}
\newcommand{\subtreeweight}{\textbf{\textproc{SubtreeWeight}}}

\newcommand{\pathweight}{\textbf{\textproc{PathWeight}}}
\newcommand{\lcaop}{\textbf{\textproc{LCA}}}
\newcommand{\diameterop}{\textbf{\textproc{Diameter}}}
\newcommand{\centerop}{\textbf{\textproc{Center}}}
\newcommand{\markop}{\textbf{\textproc{Mark}}}
\newcommand{\unmarkop}{\textbf{\textproc{Unmark}}}
\newcommand{\nearestop}{\textbf{\textproc{NearestMarked}}}

\newcommand{\tcbuild}{\textbf{$\textproc{Build}$}}
\newcommand{\tcaddvertices}{\textbf{$\textproc{AddVertices}$}}
\newcommand{\tcremovevertices}{\textbf{$\textproc{RemoveVertices}$}}
\newcommand{\tcbatchinsert}{\textbf{$\textproc{BatchLink}$}}
\newcommand{\tcbatchdelete}{\textbf{$\textproc{BatchCut}$}}
\newcommand{\tcbatchisconnected}{\textbf{$\textproc{BatchConnected}$}}
\newcommand{\tcbatchsubtree}{\textbf{$\textproc{BatchSubtreeWeight}$}}
\newcommand{\tcbatchpath}{\textbf{$\textproc{BatchPathWeight}$}}
\newcommand{\tcbatchpathmin}{\textbf{$\textproc{BatchPathMin}$}}
\newcommand{\tcbatchpathsum}{\textbf{$\textproc{BatchPathSum}$}}
\newcommand{\tcbatchlca}{\textbf{$\textproc{BatchLCA}$}}
\newcommand{\tcfindrep}{\textbf{$\textproc{FindRepr}$}}
\newcommand{\tcisconnected}{\textbf{$\textproc{IsConnected}$}}
\newcommand{\tcbatchmark}{\textbf{\textproc{BatchMark}}}
\newcommand{\tcbatchunmark}{\textbf{\textproc{BatchUnmark}}}
\newcommand{\tcbatchnearest}{\textbf{\textproc{BatchNearestMarked}}}

\providecommand{\set}[1]{\ensuremath{\left\{#1\right\}}}

\section{Batch Queries on RC-Trees}\label{sec:batch-queries-on-rc-trees}

One of the advantages of parallelism is the ability to support batch queries, where the goal is to answer
a set of queries in less work than the work of answering each query individually and in parallel with low span.
RC-Trees are highly ammenable to batch query algorithms due to having a number of very useful local and global
properties. In this section we enumerate some of these useful properties and describe batch algorithms for
subtree queries, path queries, LCAs, and nearest marked vertices.

\subsection{Useful properties of RC-Trees}

\myparagraph{The Cluster Path}
Binary clusters are particularly useful because they can simultaneously be used to represent information about the entire subtree that they contain,
or to represent information purely about the path between the two boundary vertices. The latter is used quite extensively so it gets a name, the \emph{cluster path}.

\begin{definition}[Cluster path]
	Given an RC-Tree for an input forest $F$ and a binary cluster of the RC-Tree with boundary vertices $u$ and $v$, we define the
	\emph{cluster path} of the binary cluster as the path connecting $u$ and $v$ in $F$.
\end{definition}

\noindent Since the cluster path of a composite binary cluster is just the union of the cluster paths of its binary children, it is
easy to aggregate and store information about the cluster paths.

\myparagraph{The \CommonBoundary{}}
When describing the algorithms and properties for path and subtree queries, we will often encounter a special vertex, so we also give it a name. Given
any two vertices $u$ and $v$ that are connected (in the same tree) in $F$, we define the \commonboundary{} of $u$ and $v$.

\begin{definition}[\Commonboundary{}]
	Given an RC-Tree for an input forest $F$ and a pair of vertices $u$ and $v$ that are connected in $F$, the \commonboundary{}
	is the representative vertex of the LCA of the clusters $U$ and $V$, where $U$ and $V$ are the clusters
	represented by $u$ and $v$ respectively. The LCA is with respect to the root cluster $R$ of the connected component containing $u$ and $v$,
\end{definition}

\noindent The \commonboundary{} $c$ can be found by simply walking up the tree from $u$ and $v$ until they meet at their LCA. Note that
an edge case is possible where the lowest common ancestor of $U$ and $V$ is one of $U$ or $V$. This means that $u$ or $v$ is a boundary vertex of one of
the clusters containing the other one.

\myparagraph{Path Decompositions}\label{sec:path-decompositions}
The key fact that allows RC-Trees to handle path-based operations is the following \emph{path decomposition} property. Every path
in the underlying forest $F$ can be represented by a small set of cluster paths that are all located adjacent to a corresponding
path in the RC-Tree.

\begin{theorem}[Path decomposition property]\label{thm:path-decomposition-property}
	Suppose we are given an RC-Tree for an input forest $F$ on $n$ vertices and a pair of vertices $u$ and $v$ that are connected in $F$. Let $U$ be the cluster
	represented by $u$ and $V$ be the cluster represented by $v$, and let $P$ be the path connecting $u$ and $v$ in $F$. There exists a set of
	disjoint binary RC Clusters such that:
	\begin{enumerate}[leftmargin=*]
		\item the union of their cluster paths is exactly $P$,
		\item each cluster is a direct child of the path in the RC-Tree between $U$ and $V$.
	\end{enumerate}
\end{theorem}

\noindent A proof and example can be found in the supplementary material. 

\myparagraph{Subtree Decompositions}\label{sec:subtree-decomposition}
Subtree queries supported by dynamic tree data structures typically come in one of two flavors. Either the underlying tree has a particular fixed root (which can
sometimes be changed by an explicit re-root operation), or it is a free/unrooted tree. In the rooted case, subtree queries take a single vertex $u$ and ask
for the sum of the weights on the vertices or edges in the subtree rooted at $u$.  For unrooted trees, a single vertex is not enough information to define
a subtree, so a query will typically take two arguments: a root vertex for the subtree $u$, and the parent vertex $p$.  This kind of query is more
general but consequently more tricky to implement since the orientation of a subtree could change between queries. Thus, the tree cannot in advance be preprocessed with respect to a particular root. 

Similar to the path decomposition property, RC-Trees are robust enough to handle arbitrarily rooted subtree queries because of a powerful \emph{subtree decomposition}
property, which says that any rooted subtree in $F$ can be represented by a small set of clusters that are all adjacent to a corresponding path in the RC-Tree.

\begin{theorem}[Subtree decomposition property]\label{thm:subtree-decomposition-property}
	Consider an RC-Tree for an input forest $F$ on $n$ vertices and a rooted subtree $S$ of $F$ defined by a pair of vertices $u$ and $p$, such
	that $S$ is the subtree rooted at $u$ if $F$ were rooted at $p$. Let $U$ be the RC Node represented by $u$. There exists a set of RC Clusters
	such that
	\begin{enumerate}[leftmargin=*]
		\item the union of their contents is exactly $S$,
		\item each cluster is a direct child of the path in the RC-Tree from $U$ to its root cluster.
	\end{enumerate}
\end{theorem}

\noindent A proof and example can be found in the supplementary material.

\subsection{Batch Query Strategies}

Sequential/single queries on RC-Trees typically consist of aggregating some information along
a path or set of paths within the RC-Tree. Most commonly, a query will begin with a set of RC nodes and then sweep upwards towards either
their \commonboundary{}, or all the way to the root cluster, then sometimes propagate some information back down the tree along the
same paths. Parallel RC-Trees afford us the opportunity to implement batch queries in low span. Since the RC-Tree nodes involved in different queries overlap, we can answer multiple queries in work better than the work of answering individual queries.

A batch query typically consists of \emph{two kinds} of computations, which we will distinguish as \emph{bottom-up} and
\emph{top-down} computations.  Some queries use only one or the other, while others perform both, typically a bottom-up followed by
a top-down. A \emph{bottom-up} computation is one in which every cluster wants to compute some data which is a function of its children. For example,
computing the sum of the edge weights on the cluster path of a binary cluster is a bottom-up computation, because it can be implemented
by summing the edge weights of the cluster paths of its two binary children. Bottom-up computations are not performed at query time, but rather
they are stored as \emph{augmented values} on the clusters. This means that they are computed at build time and then maintained during update
operations. They will then be available precomputed for queries to utilize.

A \emph{top-down} computation is one in which every cluster wants to compute some data which is a function of its \emph{ancestors}, most
commonly its boundary vertices. Note that the boundary vertices of a cluster \emph{always} represent ancestors of that cluster,
and furthermore that one of them is guaranteed to represent the parent cluster. Top-down computations can not be efficiently stored
as augmented values since updating a cluster high in the tree could require updating all the descendants of that cluster which would
cost linear work. Instead, top-down computations are always performed at query time.  The query algorithm is responsible for identifying
the set of clusters for which the data is needed, which for $k$ queries typically consists of some $O(k)$ clusters and all of their
ancestors, i.e., at most $O\left(k\log\left(1+\frac{n}{k}\right)\right)$ clusters in total. The algorithm then traverses
the RC-Tree from the root cluster (or multiple root clusters in parallel in the case of a disconnected forest), visits
those relevant clusters, and computes the desired data from the data on the ancestors.

For algorithms that utilize both, the bottom-up and top-down computations are typically not independent. Our batch-query algorithms
will often maintain a bottom-up computation using augmented values and then proceed at query time by performing a top-down computation
that makes use of these values. This combination strategy is what allows us to derive some of the more complicated batch-query algorithms.

\subsection{Batch Connectivity Queries}

A batch-connectivity query over a forest $F$ takes as input a sequence of pairs of vertices $u, v$ and must answer for each of them whether they
are connected, i.e., whether they are contained in the same tree. That is, we want to support the following.

\begin{itemize}[leftmargin=*]
	\item \tcbatchisconnected$(\set{\set{u_1,v_1}, \ldots, \set{u_k, v_k}})$ takes an array of tuples representing
	queries. The output is an array where the $i^\text{th}$ entry is a boolean denoting
	whether vertices $u_i$ and $v_i$ are connected by a path in $F$.
\end{itemize}

\noindent Connectivity queries are arguably the
simplest kind of query since they require no augmented data to be stored on the clusters and no additional auxiliary data structures; only the structure of the RC-Tree is
needed to determine the answer. They are described by Acar et al.~\cite{acar2020batch}. Essentially, the problem of batch connectivity is reduced to, given a batch of vertices, find the representatives of their root clusters. This information can then solve the batch
connectivity problem by checking for each pair whether they have the same representative.

\begin{theorem}[Batch connectivity queries~\cite{acar2020batch}]\label{thm:batch-connectivity-performance}
	Given an RC-Tree for a forest on $n$ vertices, a batch of $k$ connectivity queries can be answered
	in $O\left(k + k\log\left(1+\frac{n}{k}\right)\right)$ work and $O(\log n)$ span.
\end{theorem}

\noindent In the remainder of this section, we describe five new batch query algorithms on RC-Trees that are substantially more complex to
batch and parallelize.

\subsection{Batch Subtree Queries} \label{sec:batch-subtree-query}

In a weighted unrooted tree, a subtree query takes a subtree root $u$ and a parent $p$ and asks for the sum of the weights of
the vertices/edges in the subtree rooted at $u$, assuming the tree is rooted such that $p$ is the parent of $u$. The weights can be
aggregated using any predefined associative and commutative operator (i.e., the weights are from a commutative semigroup), such as minimum, maximum, or sum. Weights
can be present on vertices or edges or both.

\begin{itemize}[leftmargin=*]
	\item \tcbatchsubtree{}$(\{(u_1,p_1), \ldots, (u_k, p_k)\})$ takes an array of tuples representing queries. The output is array where the $i^\text{th}$
	entry contains the sum over the commutative semigroup operation of the contents of the subtree rooted at $u_i$ relative to the parent $p_i$.
\end{itemize}

\noindent The algorithm begins with a bottom-up computation that stores on each cluster the total aggregate weight of the contents of that
cluster. This is stored as an augmented value and hence already available at query time. The key step of the batched algorithm
is the subsequent top-down computation which computes the contributions of the relevant clusters. A full description of the algorithm
can be found in the supplementary material.

\begin{theorem}[Batch subtree queries]\label{thm:batch-subtree-performance}
	An RC-Tree for a forest on $n$ vertices with weights from a commutative semigroup can be
	augmented to solve a batch of $k$ subtree sum queries in $O\left(k\log\left(1+\frac{n}{k}\right)\right)$ work and $O(\log n)$ span.
\end{theorem}

\subsection{Batch LCA Queries}\label{sec:batch-lca}

Lowest common ancestors (LCAs) are a useful subroutine for several tree and graph algorithms, and there exists a wide variety of algorithms
for the problem including parallel algorithms~\cite{SV88} and algorithms that work on dynamic trees~\cite{sleator1983data}.
Link/cut trees~\cite{sleator1983data,ST85}, Euler-tour trees~\cite{henzinger1995randomized}, and sequential RC-Trees~\cite{acar2005experimental}
are all able to solve LCA queries. However, we know of no existing algorithm that
can efficiently solve \emph{batches} of LCA queries on a dynamic tree. We will describe how to do just this using parallel RC-Trees.

\begin{itemize}[leftmargin=*]
	\item \tcbatchlca{}$(\{(u_1,v_1,r_1), \ldots, (u_k, v_k, r_k)\})$ takes an array of tuples representing queries. The output is array where the $i^\text{th}$
	entry is the LCA of $u_i$ and $v_i$ with respect to the root $r_i$.
\end{itemize}

\noindent The batch-LCA algorithm makes use of the static parallel LCA algorithm of Schiber and Vishkin~\cite{SV88}, the
static parallel level ancestor algorithm of Berkman and Vishkin~\cite{BV94} as well as several new tricks to efficiently
batch the computation of multiple LCA's in parallel. The full algorithm is deferred to the supplementary material.

\begin{theorem}[Batch LCA queries]\label{thm:batch-lca}
	Given an RC-Tree for a forest on $n$ vertices, a batch of $k$ LCA queries can be answered
	in $O\left(k + k\log\left(1+\frac{n}{k}\right)\right)$ work and $O(\log n)$ span.
\end{theorem}

\subsection{Batch Path Queries with Inverses}\label{sec:batch-invertible-path-queries}

In the sequential single-query setting, path queries on RC-Trees were defined to operate on weighted vertices or edges where the weights are combined by
any specified associative and commutative operation (formally, the weights form a commutative \emph{semigroup}). This meant that a single algorithm could handle
queries for the sum of the edge weights, the maximum/minimum edge weight, or the total weight with respect to any arbitrarily complicated associative
and commutative operation. Unfortunately, this turns out to be \emph{impossible} to do efficiently in the batch setting (we can not hope to achieve
$O\left(k\log\left(1+\frac{n}{k}\right)\right)$ work.)

Tarjan~\cite{tarjan1978complexity} shows that the problem of verifying the edges of an MST with edge weights from a semigroup requires
$\Omega((m+n)\alpha(m+n, n))$ time, where $\alpha$ is the inverse ackerman function. MST verification consists in querying for the
minimum weight edge in a given spanning tree between the endpoints of the $\Theta(m)$ non-tree edges of the graph, and checking that the non-tree edge is
no lighter than the MST edge. This can be reduced to a batch path query, and hence this implies a superlinear lower bound
on batch path queries. 

Fortunately, not all hope is lost. The MST verification problem, or more generally, the offline path query problem admits several efficient solutions
both sequentially and parallel when additional assumptions are made about the weights. For example, if the weights have an inverse (i.e., we can perform
subtraction), linear-work algorithms are plentiful. In this section, we will give an algorithm for batch path queries on RC-Trees when the weights
admit an inverse (formally, the weights form a commutative \emph{group}). This covers the case where the weights are real-valued numbers and the goal
is to compute the sum of the weights along a path.  It does not work for computing the minimum- or maximum-weight edge on a path since there is no
inverse operation for the minimum or maximum operation.

\begin{itemize}[leftmargin=*]
	\item \tcbatchpathsum{}$(\{(u_1,v_1), \ldots, (u_k, v_k)\})$ takes an array of tuples representing queries. The output is array where the $i^\text{th}$
	entry contains the sum over the commutative group operation of the weights on the path between $u_i$ and $v_i$.
\end{itemize}

\noindent The main idea is a simple and classic technique involving the use of LCAs. Suppose we are able to compute the total weight of the path from the RC-Tree
root to any other vertex in the tree. Then the total weight on the path between two arbitrary vertices $u$ and $v$ is the sum of the paths from the root
to $u$, plus the root to $v$, less twice the weight of the path from root to LCA of $u$ and $v$. Using our batch LCA algorithm from Section~\ref{sec:batch-lca},
all that remains is to be able to compute the total weight from the root to any other vertex. We describe such an algorithm in the supplementary material.

\begin{theorem}[Batch path queries over a commutative group]\label{thm:batch-path-query-group-performance}
	An RC-Tree for a forest on $n$ vertices with weights from a commutative group can be augmented to solve a batch of
	$k$ path sum queries in $O\left(k + k\log\left(1+\frac{n}{k}\right)\right)$ work and $O(\log n)$ span.
\end{theorem}

\subsection{Batch Path-Minimum/Maximum Queries}

In Section~\ref{sec:batch-invertible-path-queries}, we discussed the impossibility of a general batch path query algorithm that works for any 
commutative semigroup (i.e., any desired associative and commutative operation over the weights with no additional assumptions), and gave an
algorithm that works for any commutative group, i.e., the case where the weights admit an inverse, such as computing the sum of the weights on
a path. Notably, this doesn't cover perhaps the most well studied kinds of path queries, which are path minima and maxima queries. These
are often referred to as \emph{bottleneck} queries, and are important because they form the basis of many algorithms for MSTs and MST verification~\cite{tarjan1979applications,king1997simpler,karger1995randomized}, as well as network flow algorithms~\cite{tarjan1983data}.
In this section we show that the special case of path minima/maxima queries also circumvents the lower bound and is efficiently solvable.

\begin{itemize}[leftmargin=*]
	\item \tcbatchpathmin{}$(\{(u_1,v_1), \ldots, (u_k, v_k)\})$ takes an array of tuples representing queries. The output is array where the $i^\text{th}$
	entry is the lightest edge on the path between $u_i$ and $v_i$.
\end{itemize}

\noindent  A classic technique for solving bottleneck problems on weighted trees is to first shrink the tree to a smaller
tree such that the minimum weight edge on the paths between the query vertices is unaffected. Anderson et. al~\cite{anderson2020work} give a parallel
algorithm that achieves this. Specifically, given an RC-Tree over a weighted input tree and $k$ query vertices, their algorithm produces a so-called
\emph{compressed path tree} containing the query vertices and at most $O(k)$ additional vertices such that the path maxima between every pair of
query vertices is preserved. Using this tool, the idea is then to reduce the problem to a small static offline path minima problem and then use the
existing algorithm of King et al.~\cite{king1997optimal} to solve it. Further details are provided in the supplementary material.

\begin{theorem}[Batch path minima/maxima queries]\label{thm:batch-path-minima-performance}
	An RC-Tree for a forest on $n$ vertices with comparable edge weights can be augmented to solve a batch of
	$k$ path minima/maxima queries in $O\left(k + k\log\left(1+\frac{n}{k}\right)\right)$ work and $O(\log n)$ span.
\end{theorem}

\subsection{Batch Nearest Marked Vertex Queries}

In the sequential setting, RC-Trees have been used to solve the \emph{nearest marked vertex} problem in a non-negative edge-weighted
tree~\cite{acar2005experimental}. Vertices may be marked or unmarked by update operations, and a query must then return the nearest
marked vertex to a given vertex. Optionally, edge weight updates can also be supported.

\begin{itemize}[leftmargin=*]
	\item \tcbatchnearest{}$(\{v_1, \ldots, v_k\})$ takes an array of vertices representing queries. The output is an array where the $i^\text{th}$ entry is
	the nearest marked vertex to $v_i$.
\end{itemize}

\noindent The algorithm consists in maintaining augmented values via a bottom-up computation which
locate the nearest marked vertices to the representatives inside each cluster, then using a top-down computation to find the global nearest
marked vertices by considering those both inside and outside each cluster of interest. Full details of the algorithm are given in the
supplementary material.

\begin{theorem}[Batch nearest marked vertex queries]\label{thm:batch-marked-performance}
	An RC-Tree for a forest on $n$ vertices with non-negative edge weights can be augmented to solve a batch of
	$k$ nearest marked vertex queries in $O\left(k + k\log\left(1+\frac{n}{k}\right)\right)$ work and $O(\log n)$ span.
\end{theorem}

\section{Ternarizing arbitrary degree trees\label{sec:ternarization}}

Ideally, we want to perform batch queries on arbitrary degree trees. However, the algorithms discussed in Section \ref{sec:batch-queries-on-rc-trees} require constant degree trees. We use a technique called \emph{ternarization} to turn an arbitrary degree tree into a degree 3 tree, as illustrated in figure \ref{fig:ternarization}. However, it is not guaranteed a priori that a batch query will return the same result on the ternarized version of the tree. In this section, we discuss the ternarization process and its interaction with various batch queries. Ternarization has been used in the literature \cite{johnson1992optimal} but never implemented nor rigorously described.



\begin{figure}[h]
	\centering
	\includegraphics[width=0.6\textwidth]{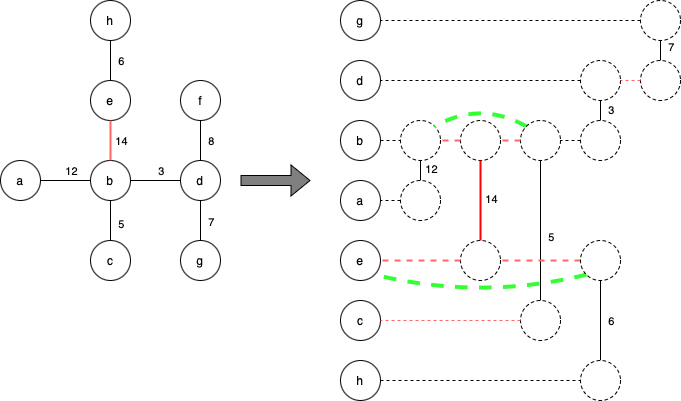}
	\caption{An arbitrary degree tree may be ternarized by inserting edges to dummy vertices with an identity weight (dotted edges). 
	An edge deleted (red) in the original graph may correspond with upto 5 delete edges (red) and 2 add edges (green).
    }
	\label{fig:ternarization}
\end{figure}

\begin{theorem}[Cost of ternarization]
	Given $k$ add and delete edges, ternarization can be performed in $O(k)$ work in expectation and $O(\log k)$ span, w.h.p. in $k$.
	\label{theorem:ternarization-cost}
\end{theorem}

\begin{theorem}[Bound on number of ternarized edges]
	Given $k$ add and delete edges, ternarization will convert these into $O(k)$ add and delete edges with constant degree.
	\label{theorem:ternarization-bound}
\end{theorem}

\begin{theorem}[Equivalence of path sums]
	The path sum between any pair of vertices in the original tree is preserved in the ternarized tree.
	\label{theorem:preserve-path-sums}
\end{theorem}

\begin{proof} Given that we know the maximum number of vertices is $n$, we need a few data structures to perform ternarization:

\begin{enumerate}
	\item A parallel hash table \cite{Gil91a}, initialized $Ln$ where $L > 1$.
	\item A list of vertices which can hold three $<$index, weight pairs$>$ as well as an \textit{owner} index. We need $3n-2$ of these vertices in a list. The first $n$ of these have their owner initialized to themselves while the rest are dummies without an owner. We will call this the vertex list.
	\item A free-list which is a list of indices, initialized as $[n, 3n-2)$.
	\item A list of indices as the tail which is initialized to $[0, n)$.
	\item A \textit{cursor} which is the index of the first dummy vertex not in use. This is set to $n$ initially.
\end{enumerate}

Given the input edge list of \textit{add} edges represented as two indices and a weight, we first apply a groupBy on the smaller index (say $v$) in the edge. 
Then, we do a prefix-sum over each group size to generate offsets within the free-list and allocate entries from the vertex list.
Each edge now has a dummy vertex whose owner we set to $v$.
We connect these edges in a linked-list with identity weights, gathering the edges using flatten.
In the parallel hash table, we then insert $(v,w)$ as the key and the index of the dummy vertex as the value. 
The leftmost dummy vertex is connected to the (previous) tail of $v$ and the (new tail) is set to the rightmost vertex for $v$.
A \textit{cursor} (initialized to $n$) is maintained and incremented with the sum of the size of all groups.

Then, using a parallel semisort, we sort according to $w$ in the input list.
We do a similar prefix-sum over the group sizes and assign another dummy vertex to each edge, returning a list of edges with identity weights and updating the tails, owners and cursor accordingly.
Now, in parallel, we look up the $(v,w)$ from the hash table and connect the two dummy indices, returning a list of edges between dummy vertices that carry the original weight between $(v,w)$.
As such, an original add edge will contribute 3 add edges after ternarization. 

We need a different approach for delete edges. 
Assuming we have a list of $k$ pairs $(v,w)$ where $v < w$, we first use the hash table to look up the dummy vertices that are connecting them.
Each of the dummy vertex (with owner $v$) will be connected to another dummy vertex with owner $w$. We mark both of these.
We gather the delete edges from all of these using flatten.
To prevent duplicate deletions, we only delete identity weighted edges if the neighbour on the right is unmarked and edges from $v$ to $w$ if $v < w$. 
As such, an incoming delete edge will, at most, contribute 5 delete edges (one real edge, and two dummy edges from the dummy lists of each of the two real vertices).
The index of the deleted nodes if written to the free list and the cursor is decremented accordingly.

New edges, however, must be added when doing deletes. 
This can be accomplished via randomized list contraction over all marked vertices.
For each section of contigious marked vertices, we connect the rightmost dummy vertex to the leftmost vertex and update the tail if necessary.
As such, each segment will contribute two add edges.

Since each step takes $O(k)$ expected work and $O(\log k)$ span w.h.p. in $k$, we require $O(k)$ work in expectation and $O(\log k)$ span w.h.p. in $k$ to insert $k$ add or delete edges, satisfying Theorem \ref{theorem:ternarization-bound}
As each add edge contributes three add edges and each delete edge contributes, at most, five delete and two add edges and no vertex (dummy or otherwise) has a degree greater than 3, Theorem \ref{theorem:ternarization-bound} also holds. 

Note that the path from $v$ to $w$ in $T'$ consists of the path in $T$, with additional dummy edges in between. Since the dummy edges have weight of identity, the sum of a path is preserved by ternarization proving Theorem \ref{theorem:preserve-path-sums}.

\end{proof}

\subsection{Batch queries with ternarization}

An ideal property of ternarization is that batch queries are preserved. Let $T$ be a (possibly) high degree tree and $T'$ its ternarized form. 

\myparagraph{Connectivity queries} Because ternarization maintains connectivity, batch connectivity queries automatically work on a ternarized tree. 

\myparagraph{Path queries}  By Theorem \ref{theorem:preserve-path-sums}, this is immediate.


\myparagraph{Subtree queries} Consider a subtree query on a root vertex $r$ and a neighbouring vertex $d$ that gives orientation (the direction giver). 
Let $r'$ be the dummy node of $r$ that connects to the dummy node $d'$ of $d$. 
Observe that the subtree sum of $r'$ in $T'$ is a sum of path sums, of real edges i.e. edges that correspond to edges in the subtree of $r$ w.r.t. $d$ (all in $r'$s subtree) and fake edges (which hold the identity). 
Since the one-hop path sums remain the same in a ternarized tree (via Theorem \ref{theorem:preserve-path-sums}), the subtree sum of $r'$ in $T'$ with direction giver $d'$ is equivalent to the sum of $r$ with direction giver in $d$ in $T$. 
We conclude the following theorem.

\begin{theorem} A subtree query on $r,d$ in $T$ is equivalent to a query on vertices $r'$,$d'$ in $T'$. 
The lookup for $r'$ and $d'$ incurs overhead with $O(1)$ expected work and $O(1)$ span w.h.p.
 \end{theorem}

\myparagraph{LCA queries} When we say LCA is preserved by ternarization, we mean that after dummy nodes are mapped back to their real node, that $LCA_T(u,v,r)$ is $LCA_{T'}(u,v,r)$. LCA is more difficult than other subtree and path queries because the LCA is not calculated via a repeated function application.

When reasoning about LCA and ternarization, it is more natural to use the folklore definition of LCA as the distance minimizing vertex, as described in Definition \ref{distDef2}. This definition is equivalent to the ancestor-based Definition \ref{ancDef} given earlier. We will also need Lemma \ref{zeroEdge2}. For proofs of the equivalence of the LCA definitions and of Lemma \ref{zeroEdge2}, see the supplementary material. The rough intuition behind Lemma \ref{zeroEdge2} is that moving away from $c'$ moves away from two of $\{u,v,r\}$ cannot reduce the sum of distances and so, viewing LCA through the lens of Definition \ref{distDef2}, cannot change the LCA. 

\begin{definition} \label{distDef2} Define $d(u,v)$ as the unweighted path distance from $u$ to $v$ and define $D_{u,v,r}(c):=d(u,c) + d(v,c) + d(r,c)$. Define $LCA(u,v,r') = argmin_{c \in V} D_{u,v,r}(c)$. For the forest extension of LCA, define $d(u,v)=\infty$ if $u$ and $v$ are not connected. \end{definition}

\begin{lemma} \label{zeroEdge2} Suppose $c=LCA_T(u,v,r)$. Give the edges in $E' \subseteq T$ weight 0, and all other edges weight 1. For $a,b \in T$, let $w(a,b)$ be the sum of the weight of the path from $a$ to $b$. Let $c'=argmin_{x \in T} W_{u,v,r}(x)$. Then $w(c,c')=0$. \end{lemma}

\begin{theorem} \label{tern} The owner of the LCA in a ternarized tree $T'$ is equal to the LCA in the original (high degree) tree $T$.  \end{theorem}

\begin{proof} Place all of the vertices in rows, placing one real vertex and all of its associated dummy vertices in each row. There are two types of edges: horizontal (within a row) and vertical (between rows). Note that the vertical edges are real tree edges and the horizontal edges are dummy edges. If we condensed an entire row into a single vertex, we would recover $T$. Let $w$ be the unweighted vertical distance between two vertices (no horizontal edges).

Let $u,v,r$ be real vertices, and suppose we want to find $LCA_T(u,v,r)$. Let $c'$ be the LCA in $T'$ (could be a dummy node). Let $c$ minimize $W_{u,v,r}$. By Lemma \ref{zeroEdge2} we have that $w(c,c')=0$ ($c$ and $c'$ must be on the same row). Since $d=w$ on $T$, on $T$ the real node associated with $c$ minimizes $D_{u,v,r}$. Thus by Definition \ref{distDef2}, the real node for $c$ is the LCA in $T$. Since $c'$ and the real node for $c$ are on the same row they map to the same real vertex. Thus ternarization preserves the LCA. \end{proof}

\section{Implementation details}

\subsection{Data structures}


We have two, key, connected data structures. The first is the actual RC tree itself which is capable of facilitating all queries
discussed in Section \ref{sec:batch-queries-on-rc-trees}. 
Each vertex has a corresponding \textit{cluster} struct. Linked to this is another data structure
which records the history of the original tree across rounds of contractions. 

Here are the key fields of a cluster:

\begin{itemize}
	\item The vertex's unique index.
	\item A state.
	\item A fixed size array of cluter pointers to its children.
	\item A cluster pointer to its parent.
	\item A \textit{data} field which contains accumulated data. This is modified upon contraction.
	\item A pointer to the first node in a linked list representing the clustering history of this vertex.
	\item A pointer to a \textit{last alive} node in the same linked list. 
\end{itemize}

The last alive node is the last node in the linked list before contraction. Note that a cluster does not have pointers to neighbours in the original tree. The clustering history consists of a linked list of \textit{nodes} corresponding to each cluster. Here are its fields:

\begin{itemize}
	\item Next and Prev pointers.
	\item A pointer to the corresponding cluster.
	\item A state (whether uncontracted, raked, binary, unary, basic edge etc)
	\item A list of pointers corresponding to other nodes in the same level.
\end{itemize}


Figure \ref{fig:Compression-example} shows an example of a basic tree contraction while Figure \ref{fig:contraction-structure} shows the corresponding data structures.
Note that all the edges in the original tree also are represented as objects of the cluster struct, acting like base edges. 
Also note that a cluster continues to be represented for one level more 
after it Unary/Nullary contracts 
while it can continue to be represented as a Binary Cluster in subsequent levels as long as its boundary vertices are live.

\begin{figure}[h!]
	\includegraphics[width=0.3\linewidth]{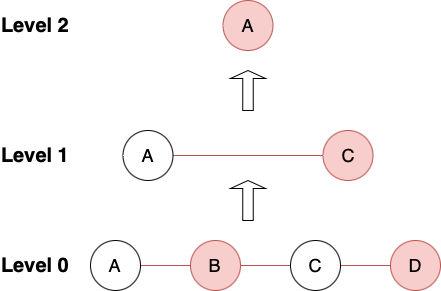}
	\caption{A simple example of contraction. In level 0, \textit{B} compresses while \textit{D} rakes into C.
		In level 1, \textit{C} rake compresses into \textit{A} using \textit{B} as the edge.
		Finally, \textit{A} nullary contracts.}
	\label{fig:Compression-example}
\end{figure}

\begin{figure}[h!]
	\includegraphics[width=0.55\linewidth]{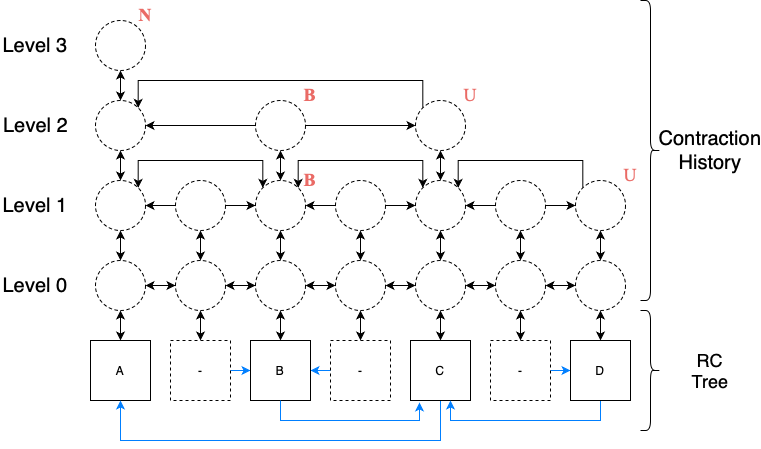}
	\caption{An illustration of our Dynamic RC tree data structure for the contraction shown in figure \ref{fig:Compression-example}. It is divided into two main parts, 
		the history that maintains the state of the tree across contraction rounds and 
		the RC tree itself. The arrows represent pointers. The squares are clusters while the circles correspond to vertices in the original tree as it contracted.
		Filled lines represent statically allocated objects while dotted lines represent dynamically allocated objects.
		The lettes in red (U/N/B) represent the state of a node as Unary, Nullary or Binary compressed. The blue lines indicate parents.}
	\label{fig:contraction-structure}
\end{figure}

\subsection{Static tree contraction} \label{sec:static-tree-contraction-practice}

Our static tree algorithm accepts a tree in the form of a list of weighted edges. This can only accept trees which have a maximum degree of 3. However, as we discuss in Section \ref{sec:ternarization}, we can represent arbitrary trees as degree 3 trees.
We use the algorithm described in \cite{anderson2024}.

Initially, we allocate a sequence of clusters, with each cluster corresponding to a vertex in the original tree. We then dynamically allocate base edge clusters, which can be done via a thread local type allocator.
Each edge is allocated once, by the thread operating on the vertex with the lower index.
Since the tree has a fixed degree, this takes $O(n)$ work and $O(\log n)$ span. 

We then create the zeroeth level of the clustering history, by allocating a node for each cluster. 
This is the first node in the corresponding cluster's linked list.
As before, the nodes for the base edge clusters are allocated by the thread operating on the lower index vertex.
However, in addition, the nodes at the zeroeth level must now point to each other -- creating the original, uncontracted tree. To this end, for each edge in the original tree, we first connect one side  of the the edge in parallel (corresponding to the vertex with the smaller index) and then the other in a separate pass.
This still takes $O(n)$ work and $O(\log n)$ span as the tree is constant degree. We mark all of these nodes as \textit{live}.

In order to contract live nodes, we first \textit{recreate} them by constructing a copy of the tree they represent in a new level.
Then, we find an MIS of \textit{eligible} nodes with a degree less than or equal to 2 as described in sec \ref{sec:coloring-mis}. 
We contract the MIS set of these nodes in place, changing the edgelist of itself and its neighbours. 
When a node contracts, we mark its state as not live and contracted (i.e. Unary, Binary or Nullary).
In addition, we set a pointer to its children in its corresponding children. 
We can also do \textit{accumulation} if necessary. For example, we can find the sum of all children if we later need to do a subtree query.
We then filter out dead nodes and recreate the next level with just the live nodes.
Note that a node that compressed via a rake (i.e. is a binary cluster) will continue to be recreated as it is treated no differently than a base edge.

Since the contraction results in a geometric decrease in the number of live nodes per level, this takes up $O(n)$ space, needs $O(n)$ work and needs $\log^2(n)$ span. 

\subsection{Dynamic tree contraction} \label{sec:dynamic-tree-contraction-practice}

Our dynamic tree contraction algorithm uses the contraction history created during the static RC tree generation to facilitate batch insertion of $k$ add or delete edges in
$O(k \log(1+\frac{n}{k}))$ work and $\log^2(n)$ span. Again, we use the algorithm described in \cite{anderson2020work}.

We first mark the nodes in the first level that correspond to the endpoints of these edges as ``affected". 
In order to facilitate maximal contraction, we also need to keep track of when a node's adjacency list is changed or when a node could have contracted in the original tree were it not for an affected node in a given layer as described in \cite{anderson2020work}. 
We mark these nodes as affected as well, by looking at nodes that meet these criteria in 2-hop neighbourhood of the endpoints.
In order to gather all the affected nodes, we use tiebreaks on each affected node -- each thread attempts to write its ID on each node's corresponding cluster.
In a subsequent pass, if the ID matches the thread's ID, then the nodes are gathered, otherwise they're filtered out.
Since each node can have a constant degree, the maximum number of two hop neighbours is also constant.
Thus, we can do the collection of affected nodes in $O(k)$ work and $O(\log k)$ span.

For delete edges, we delete the entire contraction history corresponding to the base edge that was removed. 
Note that this may result in dangling pointers in higher levels. 
However, we are careful not to dereference any such pointers.

We then use the contraction history first constructed during the static tree contraction. 
The first step is to \textit{uncontract} each affected node. 
Since we know the state of the node (whether it was live, unary, binary or nullary) in the next level, 
we can recreate the next level with just these nodes uncontracted. Nodes that remain live do not change.
However, contracted nodes become live again -- they reverse the changes caused to the adjacency lists of their neighbours in the next level.
This might affected the neighbour's adjacency list making them affected in the next level.

We can then contract an MIS set of eligible nodes as during the static tree contraction. 
In order to maintain the set of affected nodes, we gather affected nodes in a two hop neighbourhood being careful not to gather duplicates via tiebreaks as before.

We repeat this until we run out of affected nodes. 
Since each node can only contract once and can affect a constant number of nodes in each level, the maximum number of nodes that need to be modified is $O(k \log(1+\frac{n}{k}))$. 
In addition, accumulation and contraction takes constant work, the total work done is $O(k \log(1+\frac{n}{k}))$. The span is $O(\log^2(n))$ as there are $O(\log(n))$ levels with $O(\log(n))$ span each.
Note that this maintains the contraction history, allowing us to do subsequent adds and deletes to the tree.

\subsection{Path queries} \label{sec:path-queries-practice}

We implement path queries in $O(\log n)$ work and span. Since boundary vertices can always be obtained by looking at the neighbours of a cluster before it contracted (via the \textit{last alive} pointer), our RC tree structure contains all the information we need.
For a path query from a vertex $v$ to a vertex $w$, we start with the corresponding clusters for $v$ and $w$.
Then, for the cluster with the lower \textit{height} in the RC tree, we sum the value of any children that are edges (whether base edges or binary clusters).
We then ascend with this cluster.
We stop if a cluster has no parent (i.e. $v$ and $w$ were in different forests)
or we reach the same node in which case we sum the contributions from both the paths.
When ascending from a binary cluster, we must separately track values of both the boundary vertices as only one of those values may be correct.
Since each path will touch $O(\log n)$ nodes, the work and span are both $O(\log(n))$.
Note that this requires accumulating the values of each child edge while the tree is contracting.

\subsection{Subtree queries} \label{sec:subtree-query-practice}

We also implement subtree queries in $O(\log n)$ work and span. 
Given a vertex (\textit{root}) w.r.t. a neighbour for orientation (\textit{direction giver}), we can find the weight of all edges pointing away from the direction giver from the root.
This query can be broken down into two cases:
\begin{enumerate}
	\item The direction given is a boundary vertex of the root. 
	\item The root is a boundary vertex of the direction giver.
\end{enumerate}

In case 1, we accumulate the value from all children except those which may have the direction giver as a boundary vertex.
This avoids the case in which the edge connecting the root and direction giver is a child of the root.
Then we can recursively find the subtree sum for the root's children's boundary vertices (w.r.t. the root) that are not shared with the direction giver.
This scenario falls into case 2. In case 2, we accumulate values from all children except that one that contains the root. 
Then, we recurse up the tree, finding the subtree sum of the root's boundary vertices w.r.t. the root.
Since the RC tree is constant degree, we touch at most $O(log(n))$ cluster and do a constant amount of work per node.
The work and span of this query is $O(\log n)$. Ternarization can add $O(1)$ work in expectation and $O(1)$ span (with high probability) overhead.

\subsection{Batched subtree queries}

We can utilize redundancy in subtree sums to reduce the work done when we are given a batch of roots and their corresponding direction givers, using the algorithm described in Section \ref{sec:batch-subtree-query}.
Given the roots, we first ascend the RC tree, marking each cluster.
To prevent a cluster from being marked multiple times, we maintain an atomic counter per cluster that gets incremented whenever a thread tries to ascend into it.
If the counter is non-zero when a thread shows up, it terminates execution.

After all the nodes have been marked, we can gather the RC tree roots by filtering out any clusters that don't have a parent.
Afterwards, we can calculate partial sums in a top-down manner from the RC tree roots in parallel. 
For each cluster, we write down the results of a subtree-query with respect to all of its children. 
We then recurse on each of its marked children. 
Since their parent necessarily wrote down the result w.r.t. the current node, we will not need to ascend the tree more than one level. Lastly, in order to get the final results of the batched subtree query, we can do a constant time query using the stored values from every root's parent.
Ternarization can add $O(k)$ work in expectation and $O(\log k)$ span (with high probability) overhead.

\subsection{Lowest Common Ancestor} 

As described in the supplementary material, LCA requires a static level ancestors data structure. Theoretically Berkman and Vishkin's data structure is optimal \cite{BV94}, with $O(k + k \log(1+n/k))$ work for preprocessing and constant query time in our setting. However, their data structure has a $2^{2^{28}}$ constant factor. Thus, in our implementation of LCA we use brute force table lookup, which takes $O(k \log n)$ work and constant query time. Thus the work bound for our LCA implementation is $O(k \log n)$. Note that this step is the bottleneck: it is the only step in the implementation requiring more than $O(k + k\log(1+n/k))$ work. 

We handle bottom-up and top-down queries in rounds by tree level for LCA. Within each round, we need a filter operation to condense the nodes that we are to operate on in the next round. Additionally, we need a hash map for translating between the index in the original RC tree and the index in the $O(k + k\log(1+n/k))$-sized subtree we are operating on. Thus, our implementation of LCA has span $O(\log^2 n)$. Ternarization adds $O(k)$ work in expectation and $O(\log k)$ span (with high probability) overhead.

\subsection{Incremental MSF} 

In the incremental MSF problem, we must maintain an MSF subject to batches of new, weighted edges coming in. Since finding the MSF w.r.t. to the new edges and all the edges in the original tree would not be work efficient, we need to create a \textit{compressed path tree}.
Given the endpoints of the new edges, the compressed path tree maintains the minimum number of nodes and edges to describe the weight between each of the endpoints. 
This is illustrated in Figure \ref{fig:compressed-tree-gen}.

\begin{figure}[h!]
	\includegraphics[width=0.55\linewidth]{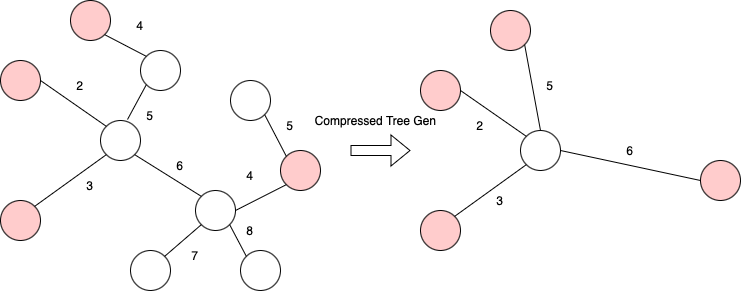}
	\caption{The compressed tree representation of a tree w.r.t. some marked nodes (red). The max between any pair of nodes is maintained in the compressed tree.}
	\label{fig:compressed-tree-gen}
\end{figure}

We construct the compressed forest gen using the algorithm in \cite{anderson2020work}. We maintain an extra set of weighted edges per cluster which correspond to represent the compressed forest. Given the set of new add edges, we first mark the endpoints. 
Then, we mark their parents as well, incrementing a \textit{per cluster} counter. 
If a thread encounters a non-zero counter, it terminates execution.

Afterwards, we gather the roots of the RC tree which contain the endpoints. We recurse to their children, creating compressed edges in a cluster local storage.
To gather the relevant vertices, we first do a top down sweep to find the number of contributions of each subtree in the RC tree.
We use these to generate offsets in a list of vertex indices where each marked cluster can write their corresponding vertex index. 
To gather the compressed edges, all of the relevant write their edges into a list if their index is the smaller of the endpoints of the compressed edge.

Once we have gathered all the compressed edges, we can append these to the list of new edges and run an MSF algorithm.
We chose to use Kruskal's because the overhead of Kruskal (despite being $O(n \log n)$ work) is minimal compared to the sparse tree generation and the dynamic insertion that follows (as shown in our evaluation).
The MSF algorithm outputs which of the original edges in the compressed tree we have to batch delete from our RC-tree and which of the add edges we have to batch insert into our RC-tree. 

Aside from Kruskal's MSF algorithm, this takes $O(k \log(1+\frac{n}{k}))$ work and $\log(n)$ span where $k$ is the number of new edges. This is because there are, at most, $O(k \log(1+\frac{n}{k}))$ marked clusters and we do a constant amount of work per cluster.
The span comes from the filters and scans.
Note that for each cluster, we need to maintain a pointer to the heaviest edge when doing tree contraction.

\subsection{Ternarization} \label{sec:ternarization-practical}

Our implementation of the ternerization framework is a blackbox, accepting a list of add edges or delete edges and returning a list of add and delete edges and changing its internal state accordingly.
We use a parlay's parallel hash table implementation under the hood for fast look-up.

\subsection{Graph Coloring and Maximal Independent Set} \label{sec:coloring-mis}
As part of deterministic tree contraction, we need to generate a Maximal Independent Set of \textit{eligible} nodes. 
A node is eligible if it has degree less than three. 
Thus, we are dealing with \textit{chains}. 
We develop a simple chain coloring scheme, we assume the tree has an initial, valid coloring (e.g. with the index).
For eligible nodes, we then assign a colour equal to the first bit that is different in their initial coloring compared to the eligible neighbour with the maximum initial color.
The exceptions are nodes which are a local maximum or a local minimum which we can assign a special colour each.
Thus we have a $O(\log n) + 2=O(\log n)$ coloring. 
This algorithm is much faster than the algorithm described in \cite{anderson2024} which has constants as high as $2^{18}$ per node for 64 bit numbers.
Using a counting sort, we can then deterministically find the MIS by iterating over each of the $O(\log n)$ colors sequentially but over each color in parallel. The only overhead is a single element in each cluster which is used to store the color.
In addition, via a flag, our ``MIS'' algorithm can return a randomly picked independent set instead. This is done by randomly colouring each node and picking only the local maximums.

\section{Methodology and Evaluation}

Our implementation is in C++ and ParlayLib \cite{blelloch2020parlaylib}.
Our experiments ran on a 48-core Amazon Web Service c7i-metal instance with
an Intel(R) Xeon(R) Platinum 8488C processor (48 cores and 3.2GHz),
and 192GB memory. 
The code was compiled with clang++14 with \texttt{-O3} and \texttt{-flto}.

\subsection{Randomized Tree Generation for Streaming}

RC trees are equipped to handles trees of arbitrary shape. As such, we designed a \textit{streaming} tree generation scheme, which offers insert and delete queries without creating a cycle.
This scheme can be used to generate add edges and delete edges for evaluation with a few different parameters.
To begin with, we connect chunks of contigious vertices as linked lists.
The length of these linked lists can be decided by a \textit{mean} parameter and a \textit{distribution}. 
The distributions include exponential, geometric, uniform and constant. 

The leftmost edge of the linked list may be connected to the linked list immediately to its left (with probability $ln$) or a random linked list anywhere on its left (with probability $1-ln$).
Varying this parameter allows us to change how deep the tree is. If $ln$ is chosen close to $1$, the tree will be very deep. An example of this is shown in Figure \ref{fig:ternarization}.

\begin{figure}
	\centering
	\includegraphics[width=0.6\textwidth]{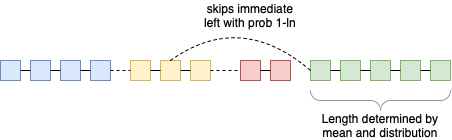}
	\caption{Our treegen scheme creates linked lists of varying lengths and connects them to vary how deep the tree is. Only the dotted edges are varied.
    }
    \vspace{-0.2in}
	\label{fig:ternarization}
\end{figure}

To generate delete edges, just the leftmost edge in each linked list is varied, allowing some structure of distinct forests to be maintained.
Lastly, every edge returned from the treegen has its vertices shuffled via a bijective map that is constructed with a parallel shuffle of $[1,n)$.

\textbf{Query generation:} Our query generation scheme is simple. For path queries, we pick $k$ random points. For subtree and batched subtree queries, we pick $k$ random points and their neighbours. For LCA, we pick $k$ random triplets.

\subsection{Results} 

The results for our static tree-gen are shown in Figure \ref{fig:static-plots}. 
The graph generation time is linear with respect to the graph size for different configurations.
The speedup is sublinear, achieving 30x speedup with 48 threads.
Furthermore, our MIS scheme performs just as well as an randomized IS scheme. 
Additionally, the depth of the tree does not affect the generation time.

The results for our batch insertion are shown in Figure \ref{fig:dynamic-plots}.
These are roughly twice as slow as constructing a static RC tree.
This is due to the overhead to deleting existing edges as well as uncontracting the next level for affected nodes. 
Note that the scenario with a mean of $1.1$ is faster as deletion of edges results in many isolated forests, bringing its speed closer to static tree generation.
The speedup remains the same as in the static case.

\begin{figure*}[t!]
    \centering
        \includegraphics[width=0.42\textwidth]{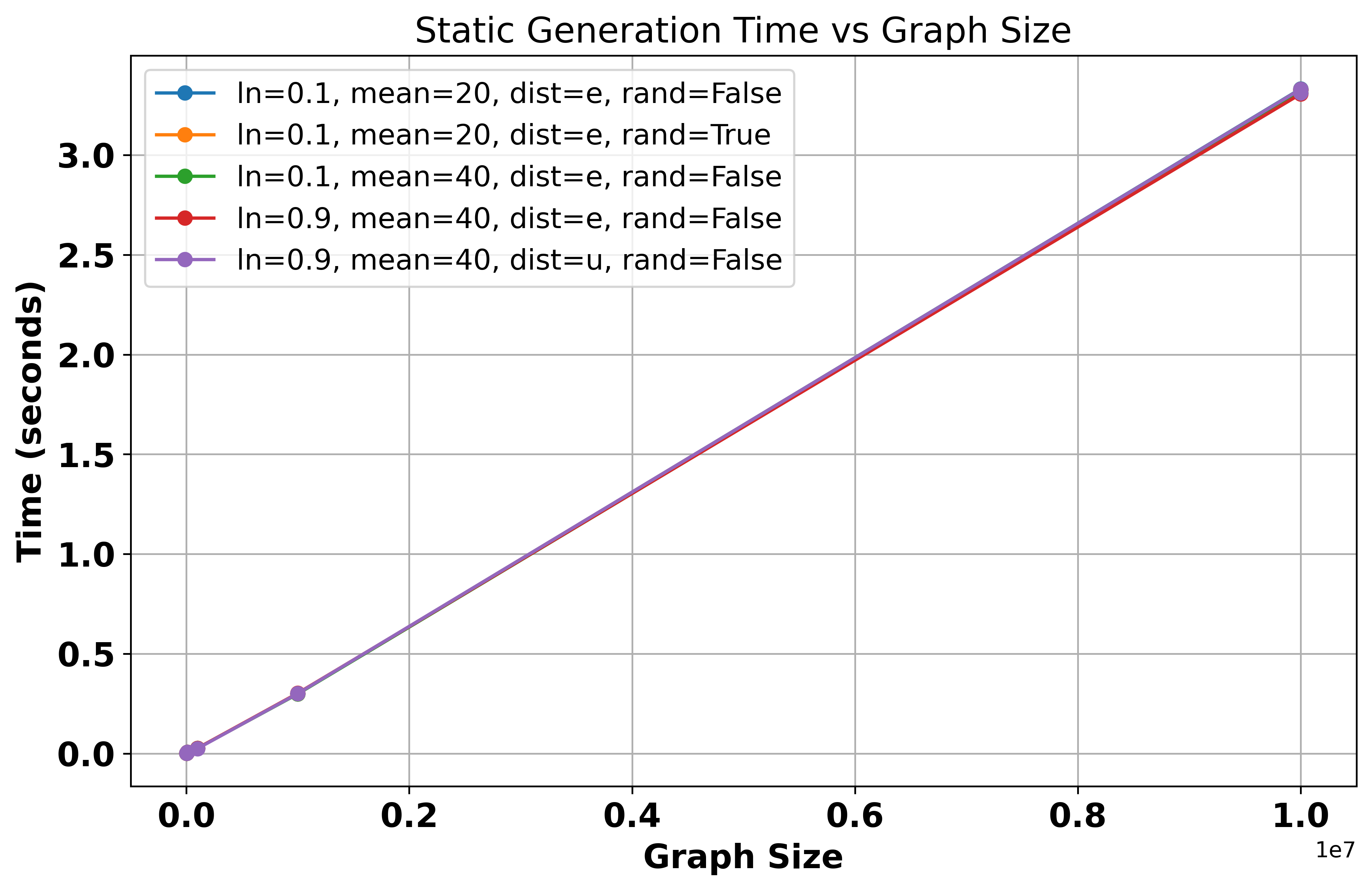}
        \includegraphics[width=0.42\textwidth]{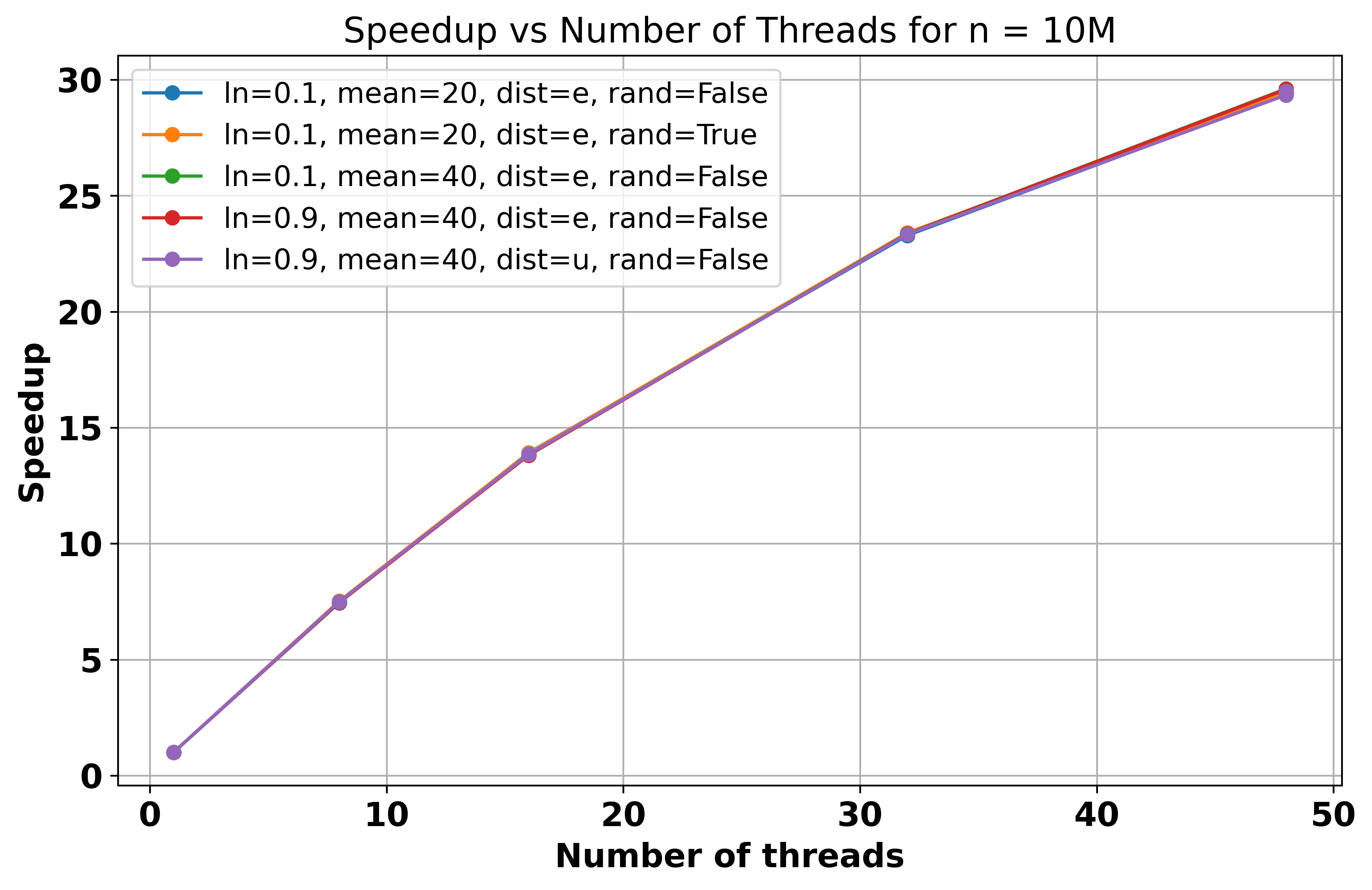}
        \caption{}
        \vspace{-0.1in}
        \label{fig:static-plots}
\end{figure*}

\begin{figure*}[t!]
    \centering
        \includegraphics[width=0.42\textwidth]{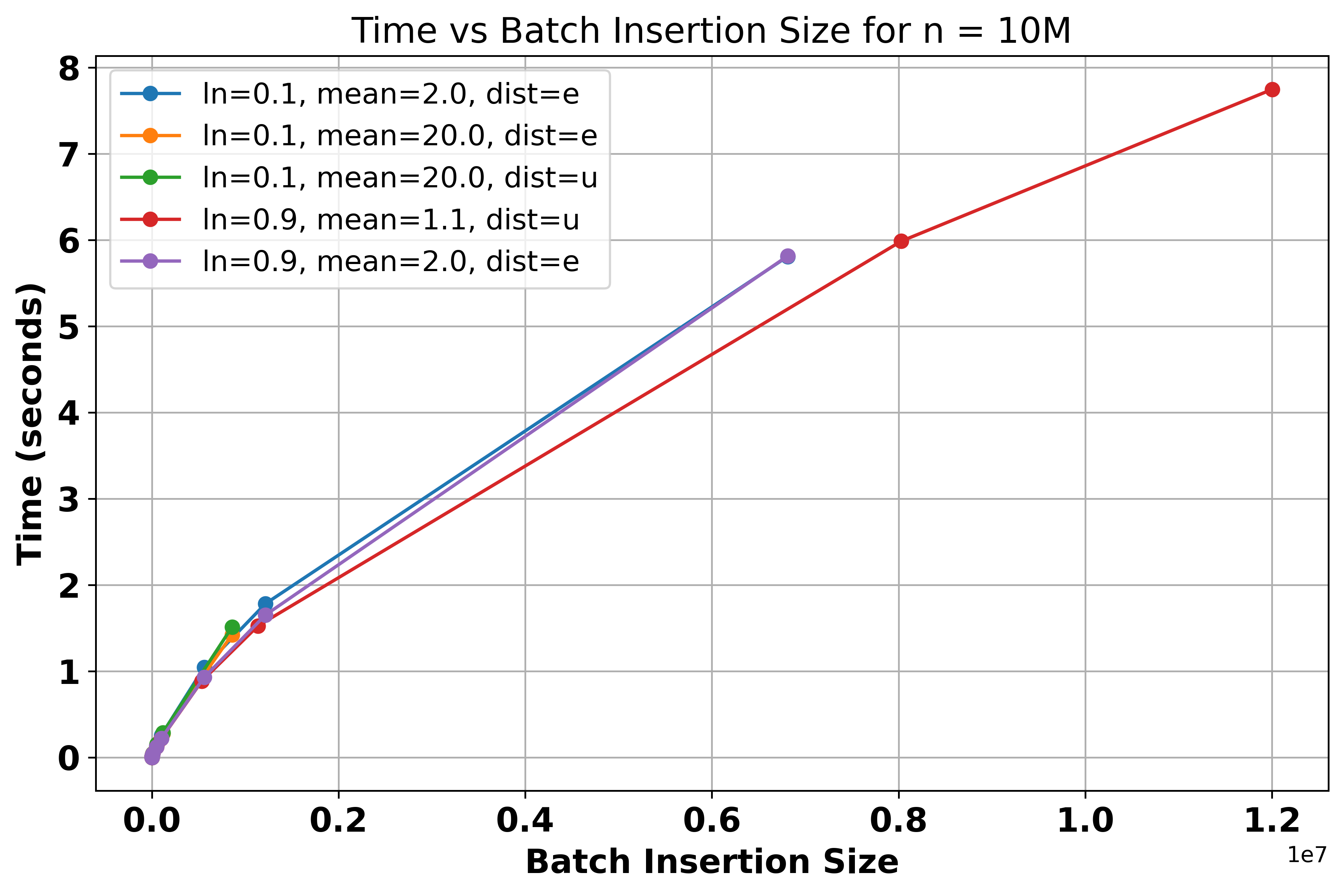}
        \includegraphics[width=0.42\textwidth]{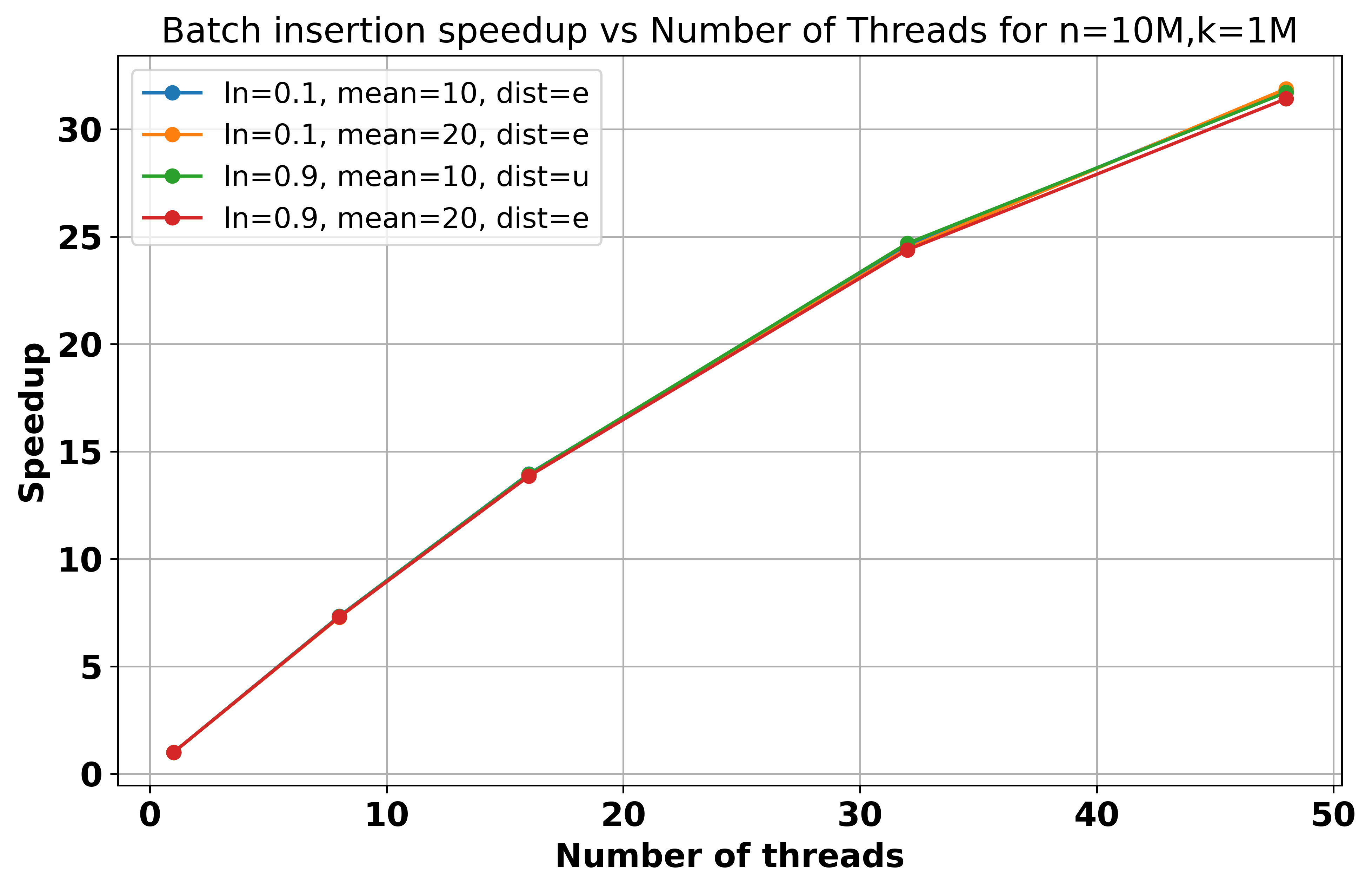}
        \caption{}
        \vspace{-0.1in}
        \label{fig:dynamic-plots}
\end{figure*}

The time for queries is shown in Figure \ref{fig:queries-time}. 
The LCA query is a an order of magnitude slower than the path and subtree queries, because each LCA batch with arbitrary roots requires 3 calls to LCA with fixed roots, which itself has multiple top-down and bottom-up sweeps.
The speedup (Figure \ref{fig:queries-speedup}) is also lower for LCA but much lower for batched subtree queries.
This is due to the atomics in batched subtree queries, which don't scale well with contention. 

\begin{figure*}[h]
    \centering
        \includegraphics[width=0.42\textwidth]{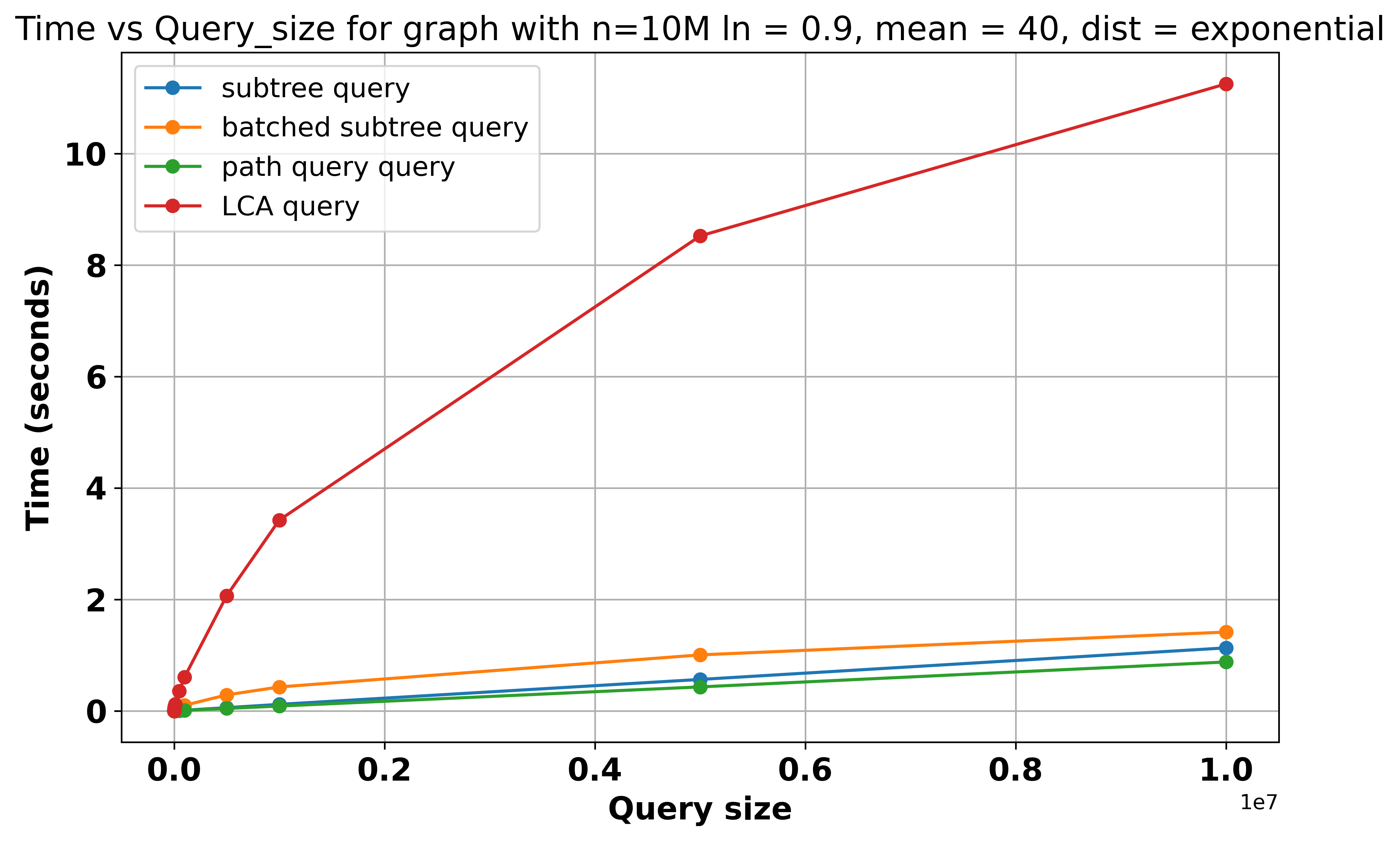}
        \includegraphics[width=0.42\textwidth]{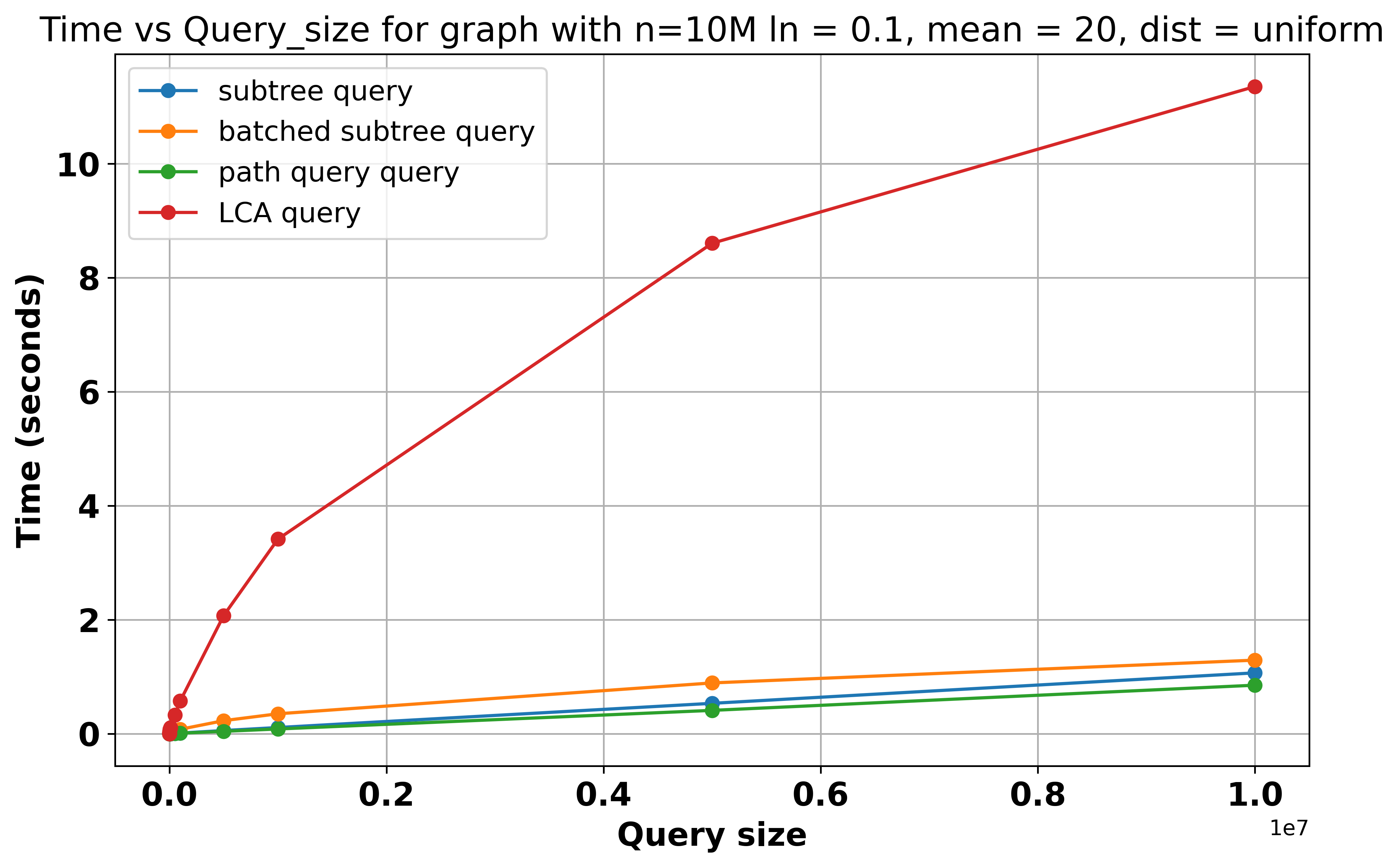}
        \caption{}
        \vspace{-0.1in}
        \label{fig:queries-time}
\end{figure*}

\begin{figure*}[h]
    \centering
        \includegraphics[width=0.42\textwidth]{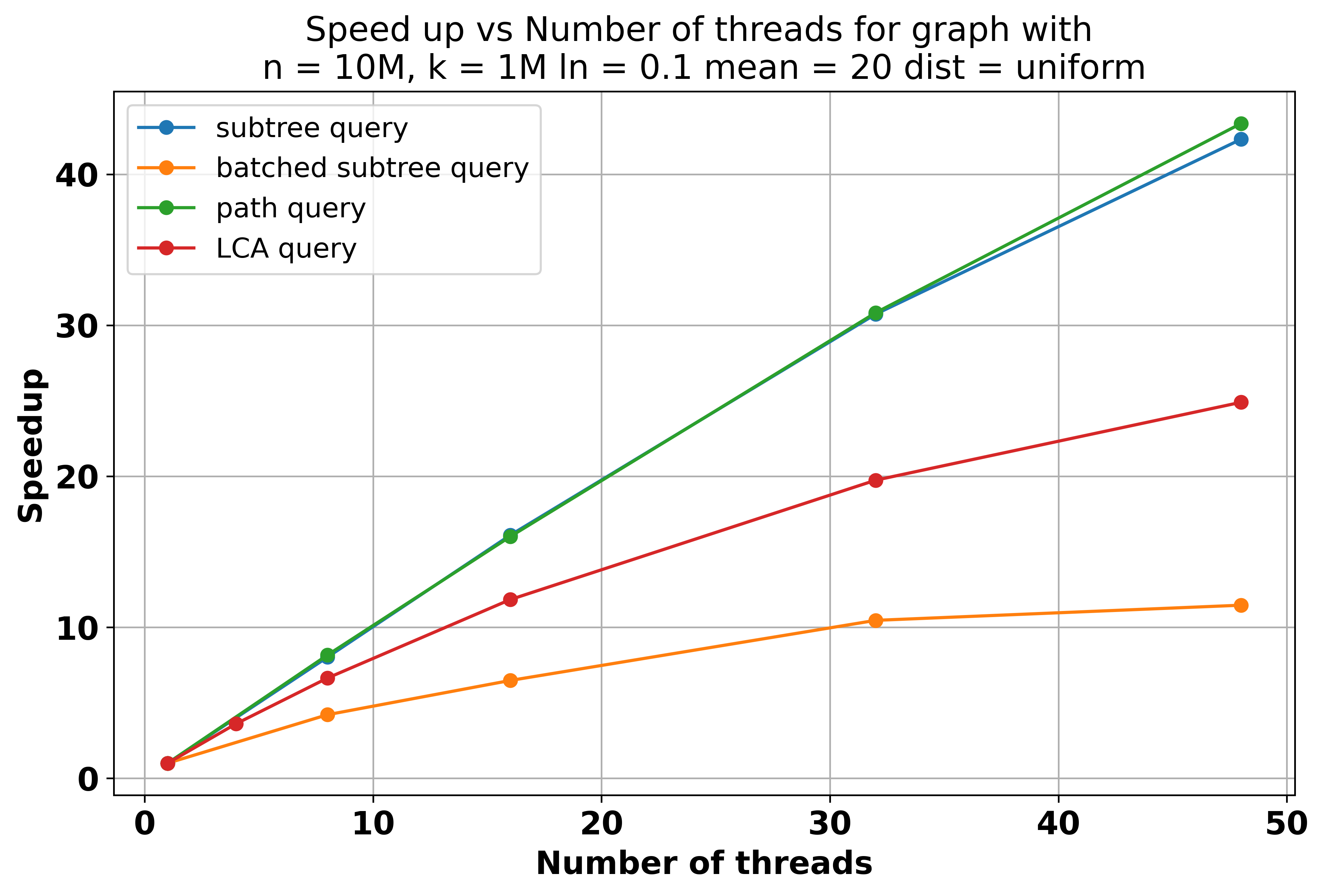}
        \includegraphics[width=0.42\textwidth]{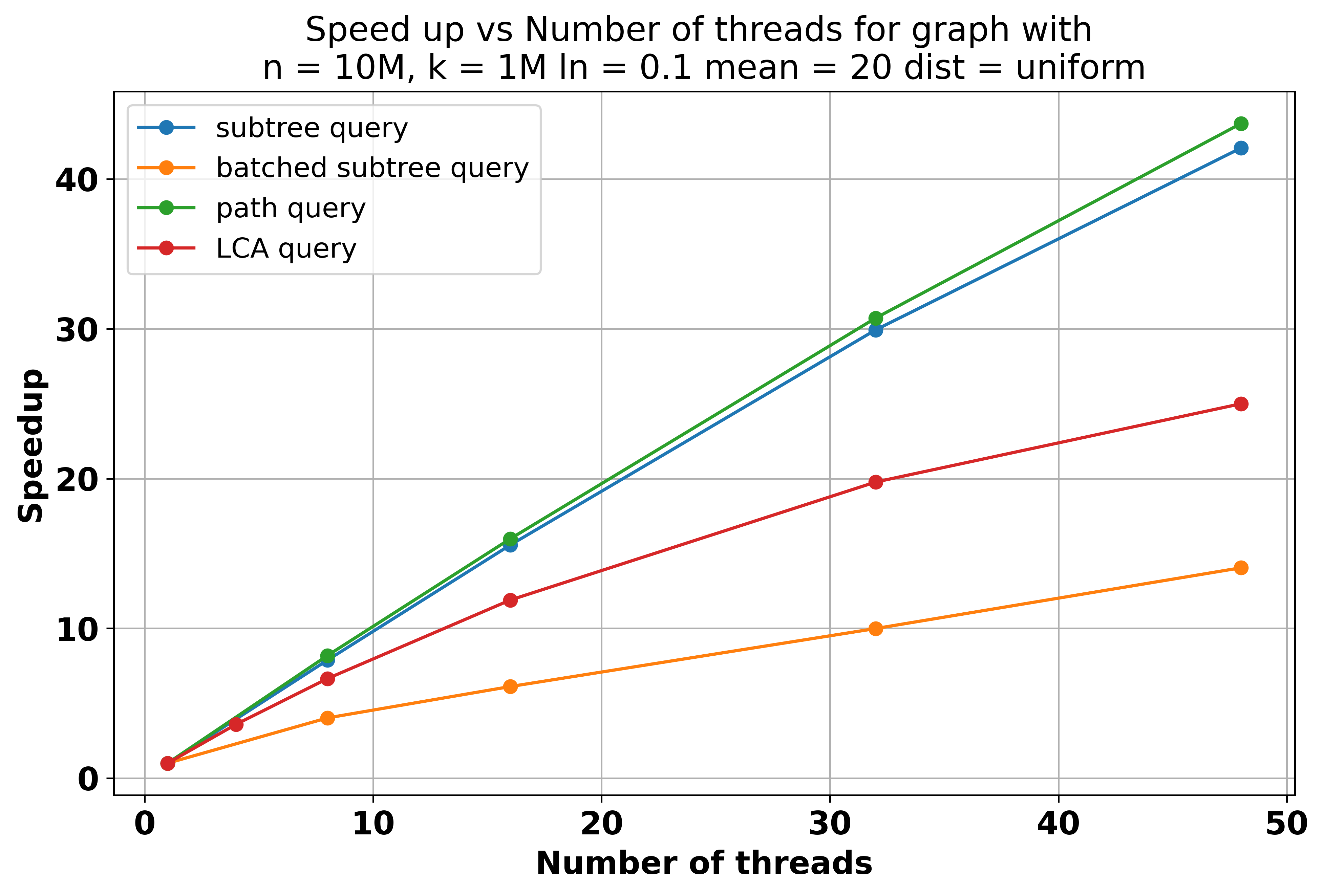}
        \caption{}
        \vspace{-0.1in}
        \label{fig:queries-speedup}
\end{figure*}

The time for the incremental MSF is shown in Figure \ref{fig:incmst}. 
We have also shown the time taken for the MSF algorithm itself, which takes minimal time.
The compressed sparse tree generation also takes about as much time as batch insertion since they have similar asymptotic work bounds.
The speed-up is higher than that of batched queries but lower than that of dynamic batch insertion -- this is because compressed tree generation uses atomics but dynamic batch insertion does not.

\begin{figure*}[t!]
    \centering
        \includegraphics[width=0.42\textwidth]{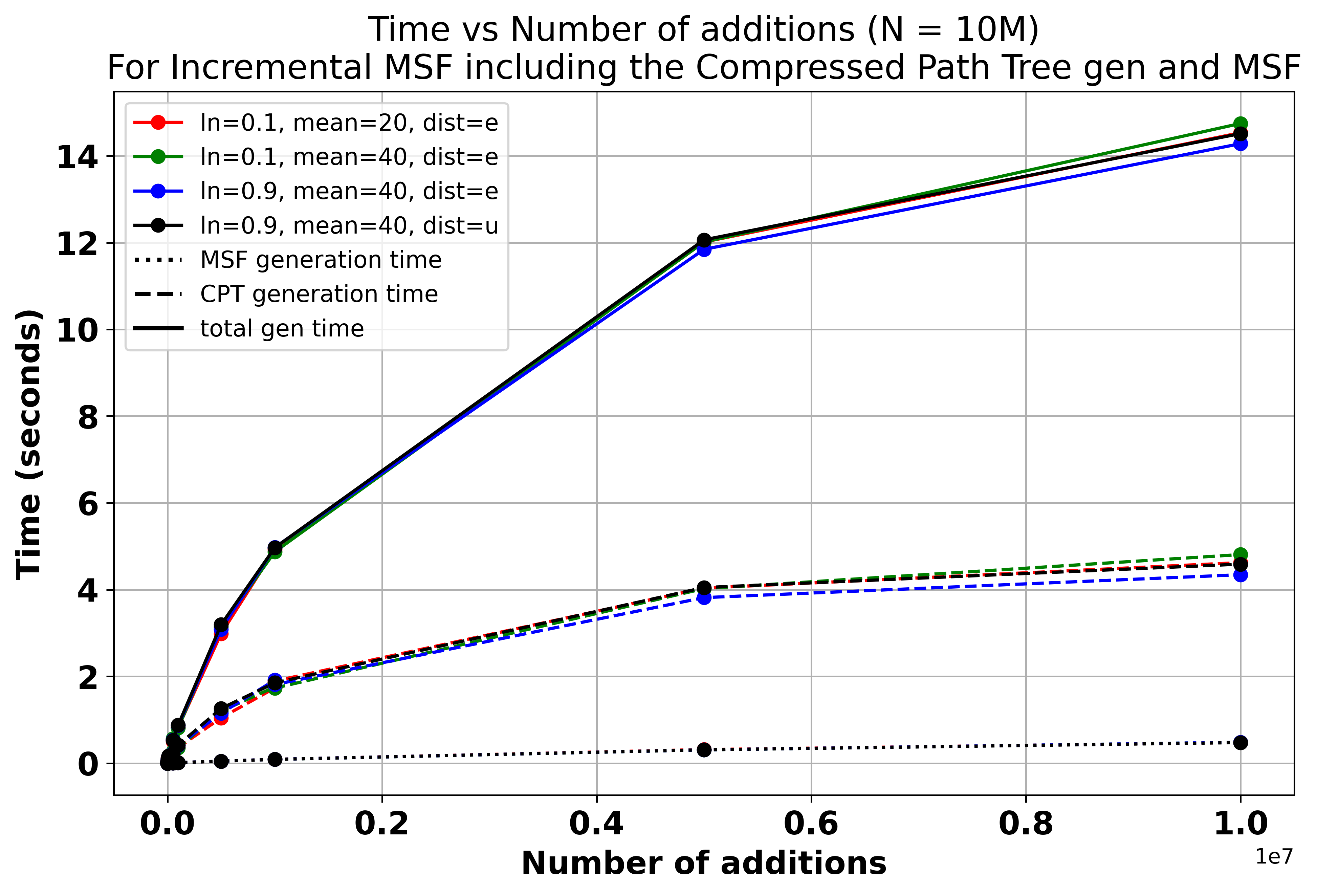}
        \includegraphics[width=0.42\textwidth]{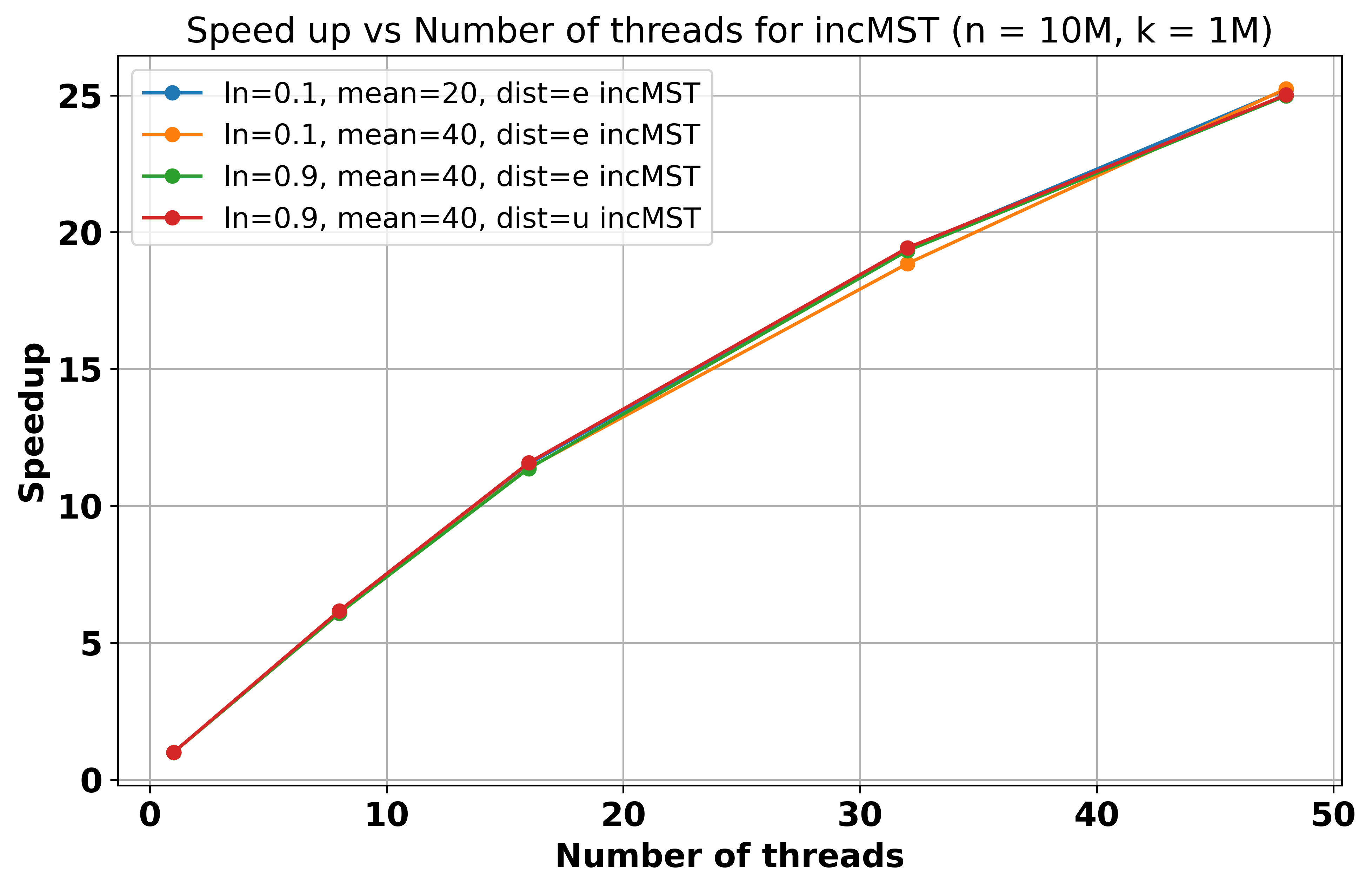}
        \caption{}
        \vspace{-0.1in}
        \label{fig:incmst}
\end{figure*}

In order to investigate subtree queries vs batched subtree queries, we ran an experiment with a graph size of 100 million.
These results are shown in Figure \ref{fig:subtree-big} -- as shown, the batched subtree query outperforms parallel subtree queries for larger values of $n$ and $k$.

\begin{figure*}[t!]
    \centering
        \includegraphics[width=0.42\textwidth]{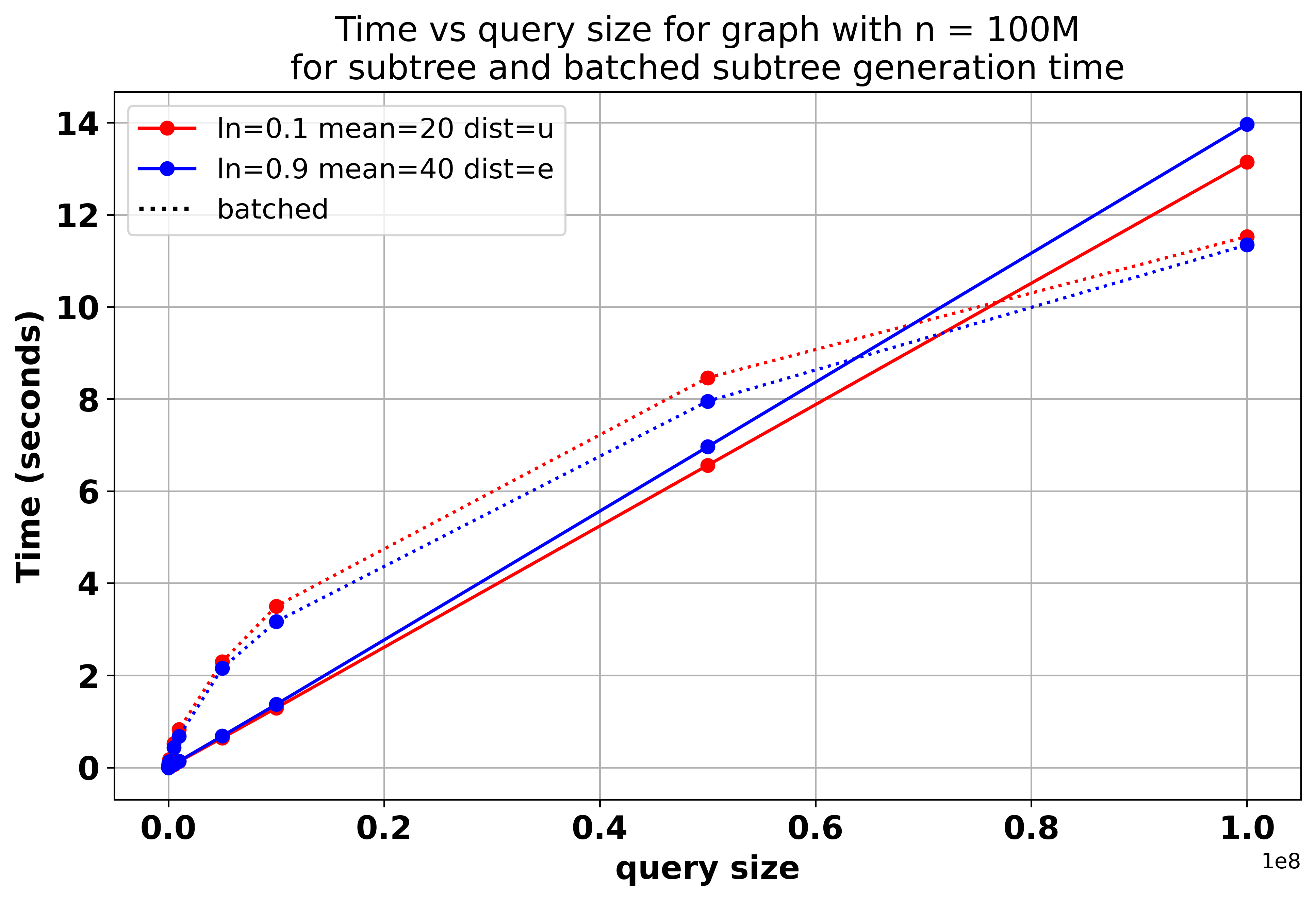}
        \caption{}
        \vspace{-0.1in}
        \label{fig:subtree-big}
\end{figure*}



\section{Conclusion}

We describe various improvements to prior work on batch dynamic RC trees including supporting batched queries in $O(k + k\log(1+\frac{n}{k}))$ work and polylog span.
We further establish a general ternarization scheme which can translate input edges into edges in a ternarized tree. 
Our ternarization scheme is work efficient (in expectation), has polylog span (with high probability) and admits batched queries. We implement this dynamic RC tree, including batch inserts and deletes, path queries, subtree queries, batched subtree queries, batched LCA queries and incremental MSF. To benchmark our implementation, we develop a random tree generator that creates trees with varying characteristics. We benchmark our implementation using this tree generator, and show that our RC tree scales well on input trees with varying characteristics. 

A next step would be to further develop the RC-tree implementation. This could involve supporting additional queries, like batch path queries and nearest-marked vertex queries. Recently, RC-trees were used to develop a min-cut algorithm with $O(m \log^2 n)$ work and polylog depth \cite{anderson2023cuts}. The min-cut result required RC-tree mixed batch updates (batch contains both tree updates and queries). One could supplement our RC-tree to support this.


\bibliographystyle{ACM-Reference-Format}
\balance
\bibliography{ref}

\appendix
\clearpage

\section*{Supplementary Material}

\section{Batch-queries on RC-Trees: Additional Information}

\subsection{The Path Decomposition Property}

\noindent An example of a path decomposition in the RC-Tree is shown in
Figure~\ref{fig:path-decomposition-example}.

\begin{figure}
	\centering
	\begin{subfigure}{0.48\textwidth}
		\centering
		\includegraphics[width=\textwidth]{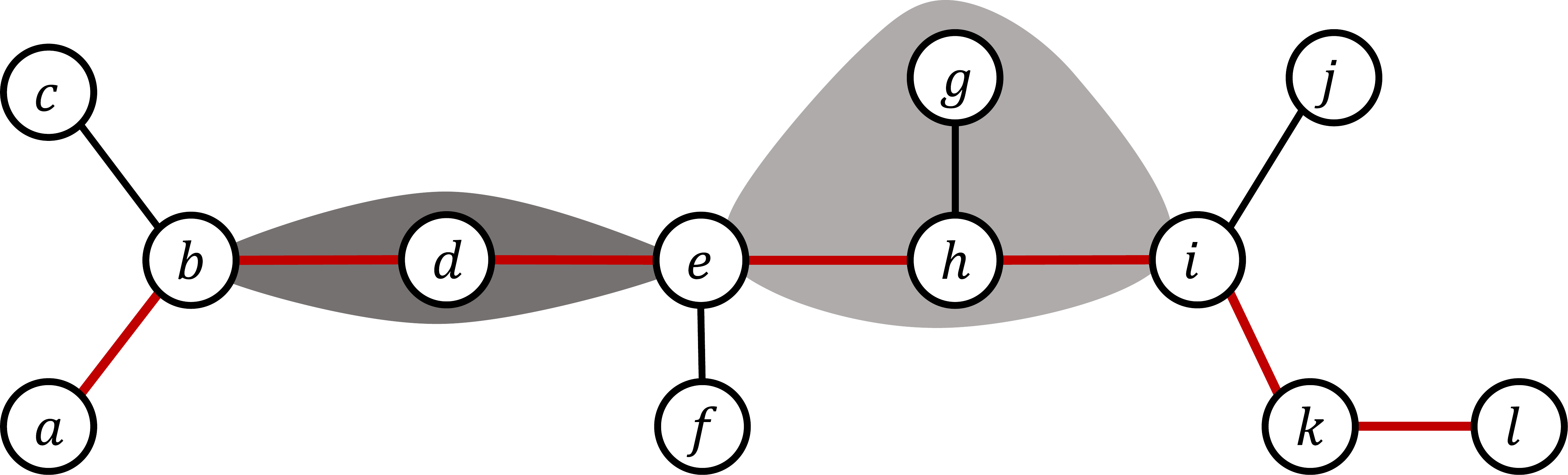}
		\caption{The path from $a$ to $l$ in the forest can be decomposed into the cluster paths of the RC clusters $(a,b), D, H, (i.k), (k,l)$.}
	\end{subfigure}\hfill
	\begin{subfigure}{0.48\textwidth}
		\centering
		\includegraphics[width=\textwidth]{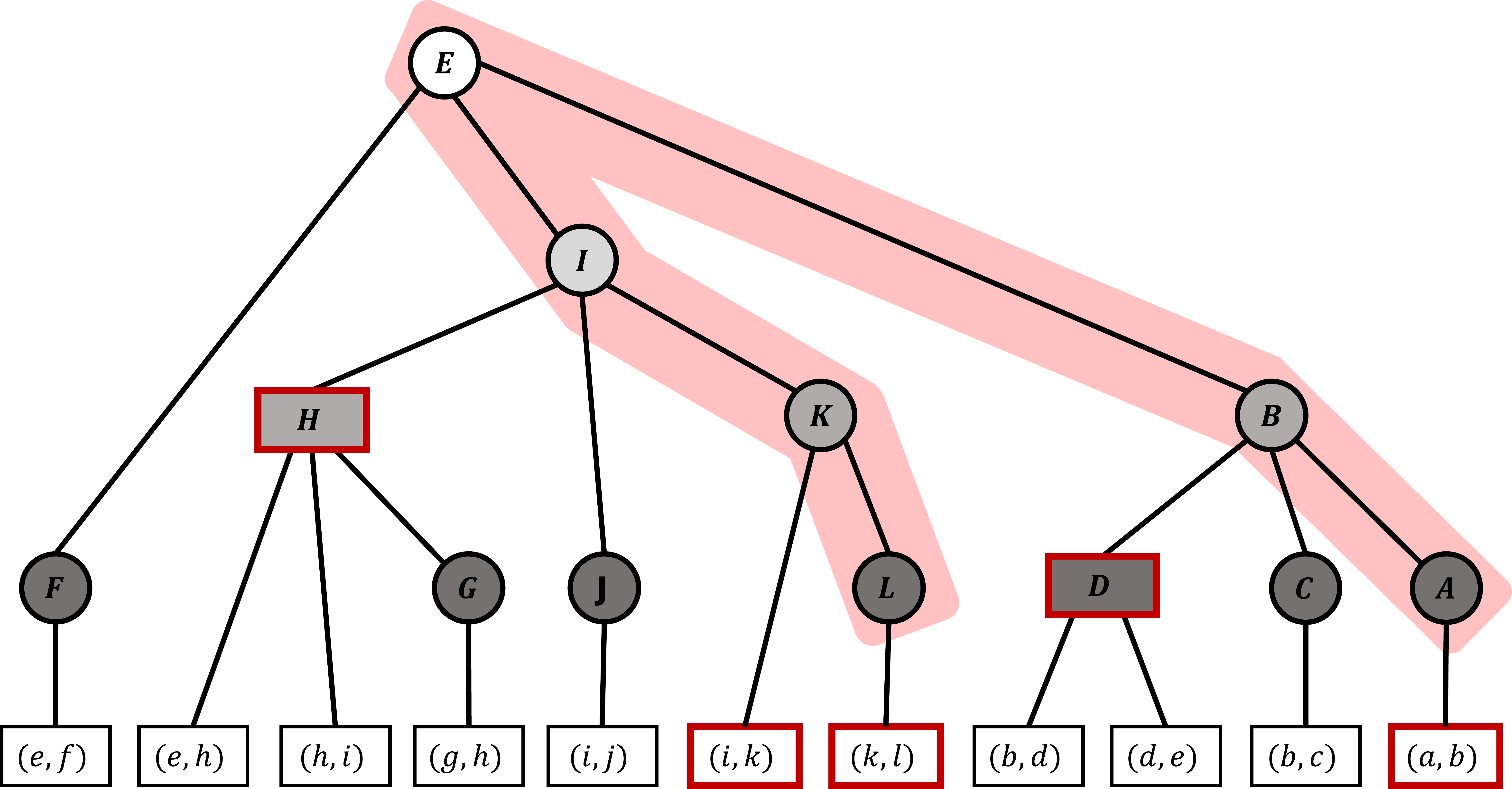}
		\caption{The clusters in the decomposition are children of the path between $A$ and $L$. The path in the RC-Tree is
			highlighted in red and the chosen clusters have a bold red outline.}
	\end{subfigure}
	\caption{An example of a path decomposition in an RC-Tree.}
	\label{fig:path-decomposition-example}
\end{figure}

\begin{proof}[Proof of Theorem~\ref{thm:path-decomposition-property}]
	We prove Theorem~\ref{thm:path-decomposition-property} constructively by describing how to find such a set of clusters. To simplify, we start by considering the \commonboundary{} $c$ of $u$ and $v$.
	Let $U, V, C$ be the clusters represented by $u, v, c$ respectively. We aim to build a path decomposition of the path from $u$ to $c$ and $v$
	to $c$, then take their union to obtain a path decomposition of $u$ to $v$. A path decomposition of $u$ to $c$ must consist of binary
	clusters that are children of the RC-Tree path from $U$ to $C$.
	
	We can show inductively that for any vertex $u$ and any cluster $B$ containing $u$ with boundary vertex $b$, there is a set of binary clusters
	that are children of $U$ to $B$ in the RC-Tree whose cluster paths form the path $u$ to $b$ in $F$. Figure~\ref{fig:path-decomposition-proof} shows
	an example of the resulting paths.
	
	Consider the first cluster to contain $u$.
	If $u$ contracts via rake then it forms some unary cluster $B$ with boundary vertex $b$. This rake operation must rake the edge $(u,b)$ which
	is represented by a binary cluster with boundary vertices $u$ and $b$ and whose cluster path is therefore the path between $u$ and $b$ in $F$.
	This cluster is a child of $B$ and hence this is a valid path decomposition as described. If instead $u$ contracts via compression and is contained
	in a binary cluster $B$, we consider either of its boundary vertices $b$. One of the two edges that compressed is $(u,b)$, which is represented
	by a binary cluster with boundaries $u$ and $b$, and hence a cluster path that covers $u$ to $b$ in $F$.
	
	Now suppose for the purpose of induction that for a cluster $B$ containing $u$ we know a path decomposition from $u$ to the boundary
	vertices of $B$. We want to show that for the parent cluster $B'$ and any of its boundary vertices $b'$ that there is a path decomposition
	from $u$ to $b'$. We will show that the new path decomposition consists of the old path decomposition with at most one additional cluster
	path added.  Consider four cases depending on whether $B$ and $B'$ are unary or binary clusters (or a combination of both).
	\begin{enumerate}[leftmargin=*]
		\item \textbf{Unary and unary}:  $b$ contracts and rakes onto $b'$ via an edge $(b,b')$ which corresponds to a binary cluster with boundaries
		$b$ and $b'$, and hence there is a path decomposition consisting of the previous path decomposition plus this cluster, which is a child of $B'$.
		
		\item \textbf{Unary and binary}: $b$ contracts and compresses the two edges $(b_1', b)$ and $(b_2', b)$, corresponding to binary clusters with boundaries
		$b$ and $b_1'$, and $b$ and $b_2'$. The resulting boundaries of $B'$ are $b_1'$ and $b_2'$, for which there are path decompositions consisting of the
		previous path decompositions plus the cluster corresponding to $(b_1', b)$ and $(b_2', b)$ respectively, each of which are children of $B'$.
		
		\item \textbf{Binary and unary}: If $B$ has boundaries $b_1$ and $b_2$, then one of them contracts and rakes the edge $(b_1, b_2)$ onto the other
		one, leaving it as the sole boundary vertex of $B'$. We therefore already know the path decomposition for this boundary.
		
		\item \textbf{Binary and binary}: If $B$ has boundaries $b_1$ and $b_2$, assume without loss of generality that $b_1$ contracts and compresses the
		two edges $(b_1, b_2)$ and $(b_1', b_1)$ resulting in an edge $(b_1', b_2)$, which corresponds to a binary cluster with boundaries $b_1'$ and $b_2$.
		A path decomposition for $b_1'$ consists of the previous path decomposition and the cluster corresponding to $(b_1', b_2)$. We already know a path decomposition of $b_2$
		since it did not change.
	\end{enumerate}
	
	\noindent The cases are illustrated in Figure~\ref{fig:path-decomposition-cases}. By induction, we can find a path decomposition from $u$ to any of its ancestors' boundaries, and hence there is a path decomposition for $u$ to $c$, the
	\commonboundary{}, and similarly for $v$ to $c$.
\end{proof}

\begin{figure}
	\centering
	\includegraphics[width=0.75\textwidth]{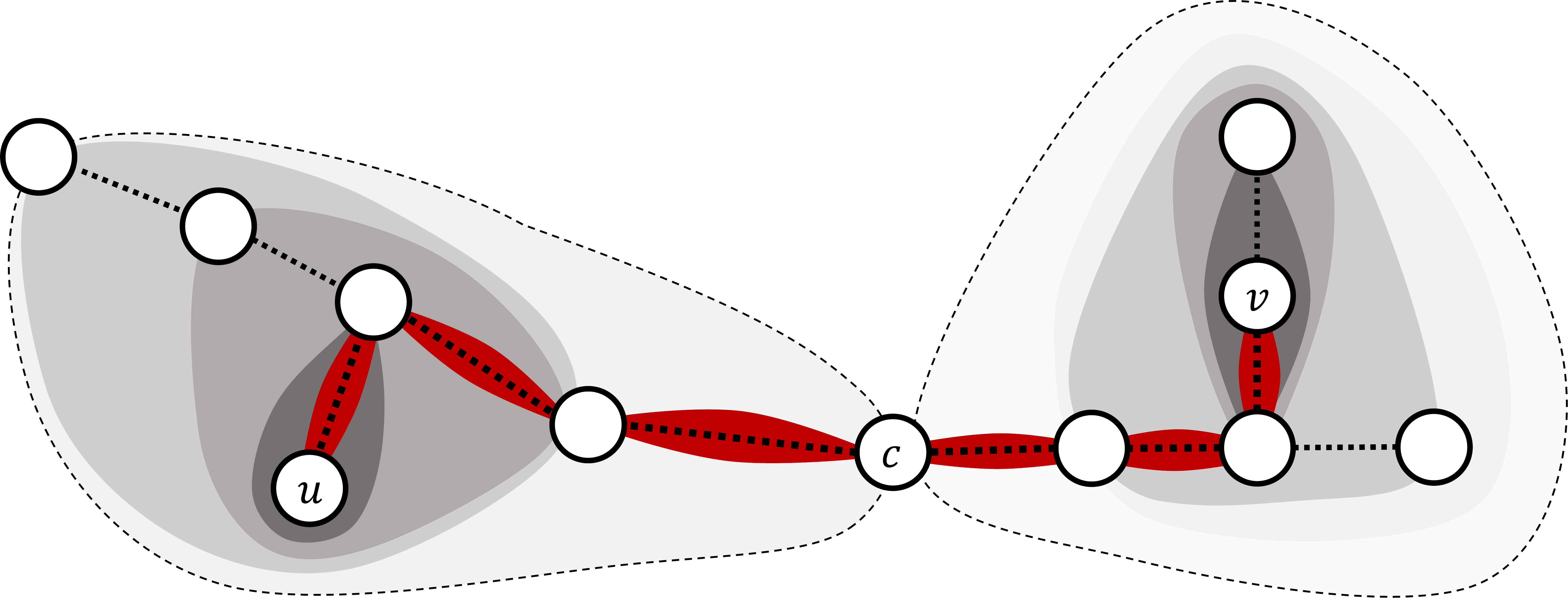}
	\caption{A path decomposition obtained by finding a set of binary clusters connecting $u$ and $v$ to their ancestors' boundary vertices until they
		meet at the \commonboundary{} $c$}\label{fig:path-decomposition-proof}
\end{figure}

\begin{figure}
	\centering
	\begin{subfigure}{0.49\textwidth}
		\centering
		\includegraphics[width=0.7\textwidth]{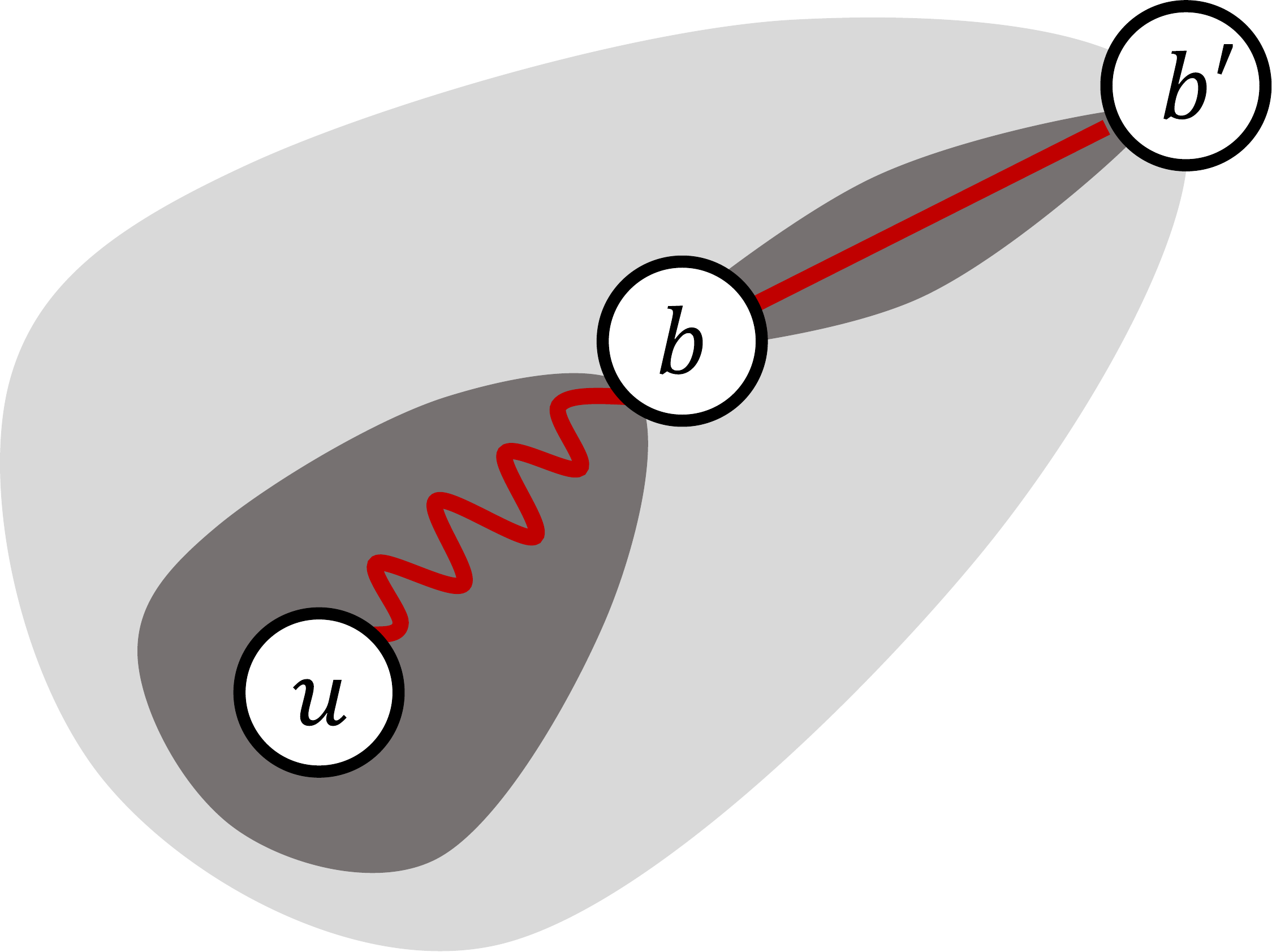}
		\caption{Case 1: A unary cluster with a unary parent.}
	\end{subfigure}
	\begin{subfigure}{0.49\textwidth}
		\centering
		\includegraphics[width=0.85\textwidth]{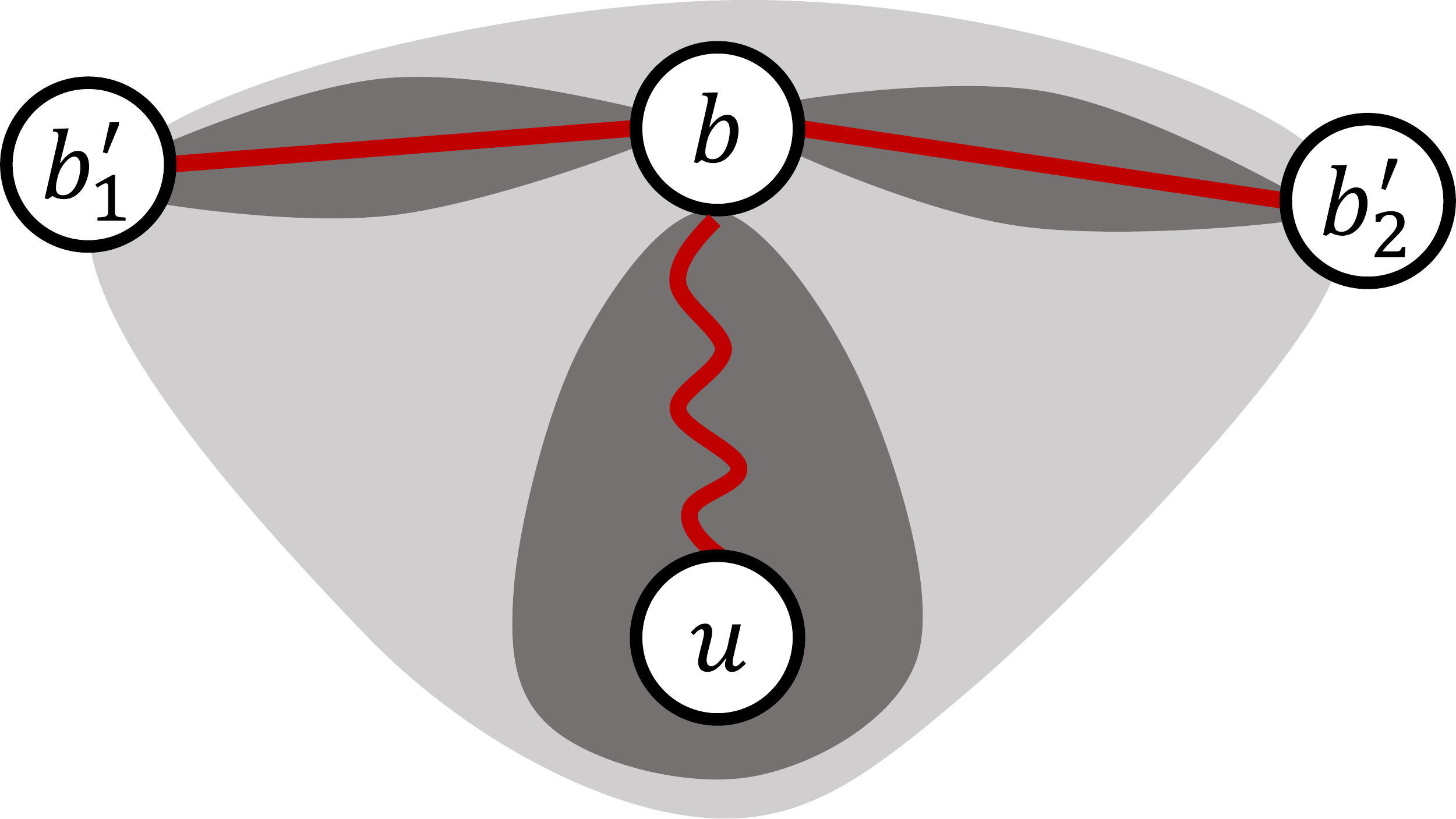}
		\caption{Case 2: A unary cluster with a binary parent.}
	\end{subfigure}
	
	\bigskip
	
	\begin{subfigure}{0.49\textwidth}
		\centering
		\includegraphics[width=0.85\textwidth]{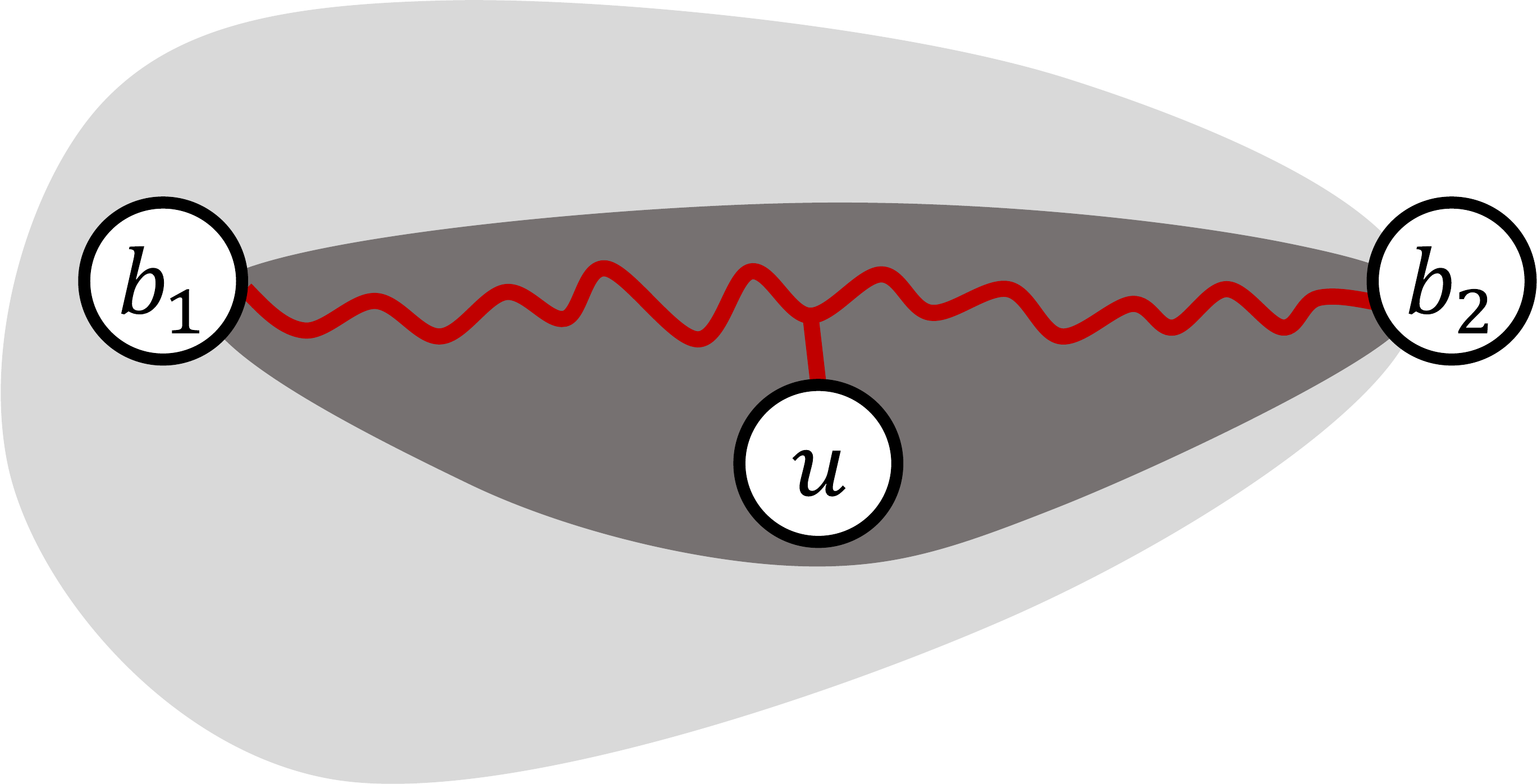}
		\caption{Case 3: A binary cluster with a unary parent.}
	\end{subfigure}
	\begin{subfigure}{0.49\textwidth}
		\centering
		\includegraphics[width=0.99\textwidth]{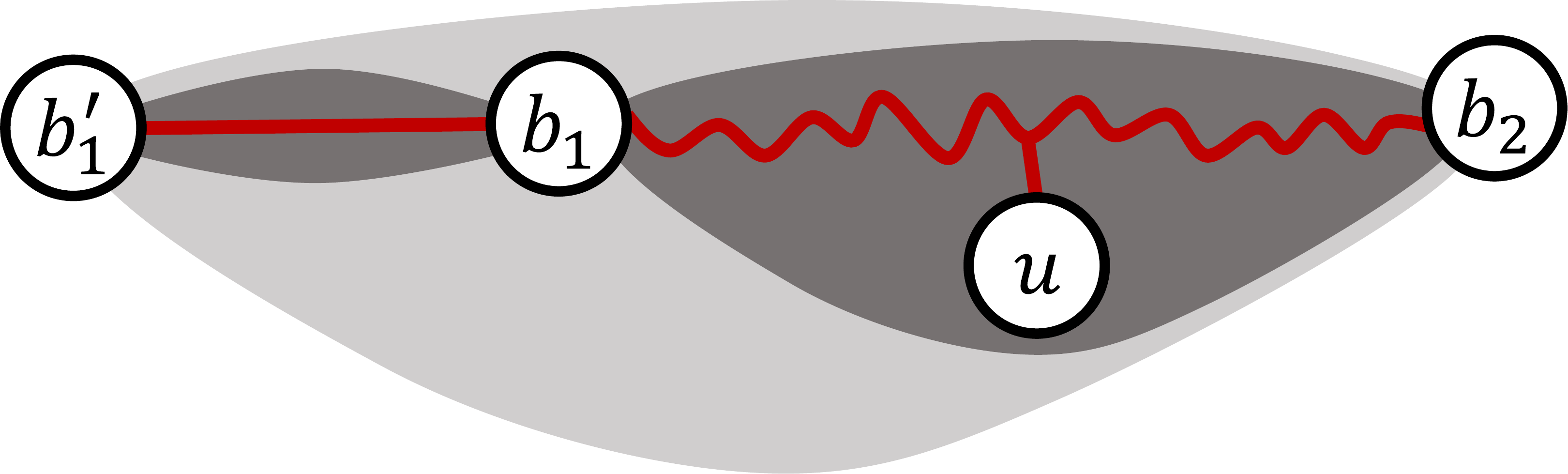}
		\caption{Case 4: A binary cluster with a binary parent.}
	\end{subfigure}
	\caption{Building a path decomposition inductively.}
	\label{fig:path-decomposition-cases}
\end{figure}

\subsection{The Subtree Decomposition Property}

\begin{figure}
	\centering
	\begin{subfigure}{0.48\textwidth}
		\centering
		\includegraphics[width=\textwidth]{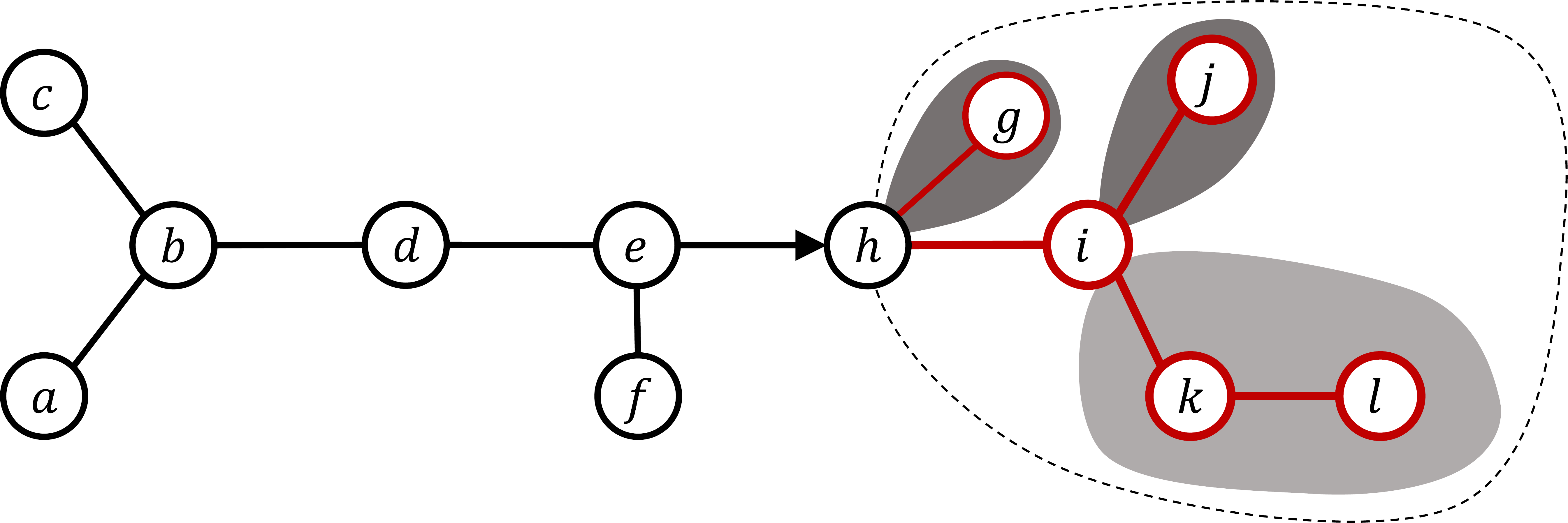}
		\caption{The subtree rooted at $h$ oriented with respect to its parent $e$ can be decomposed into the clusters $\{h \}, G, (h,i), J, K, \{ i \}$. The base
			clusters $\{ h \}$, $\{ i \}$ would be children of $H$ and $I$ if represented explicitly.}
	\end{subfigure}\hfill
	\begin{subfigure}{0.48\textwidth}
		\centering
		\includegraphics[width=\textwidth]{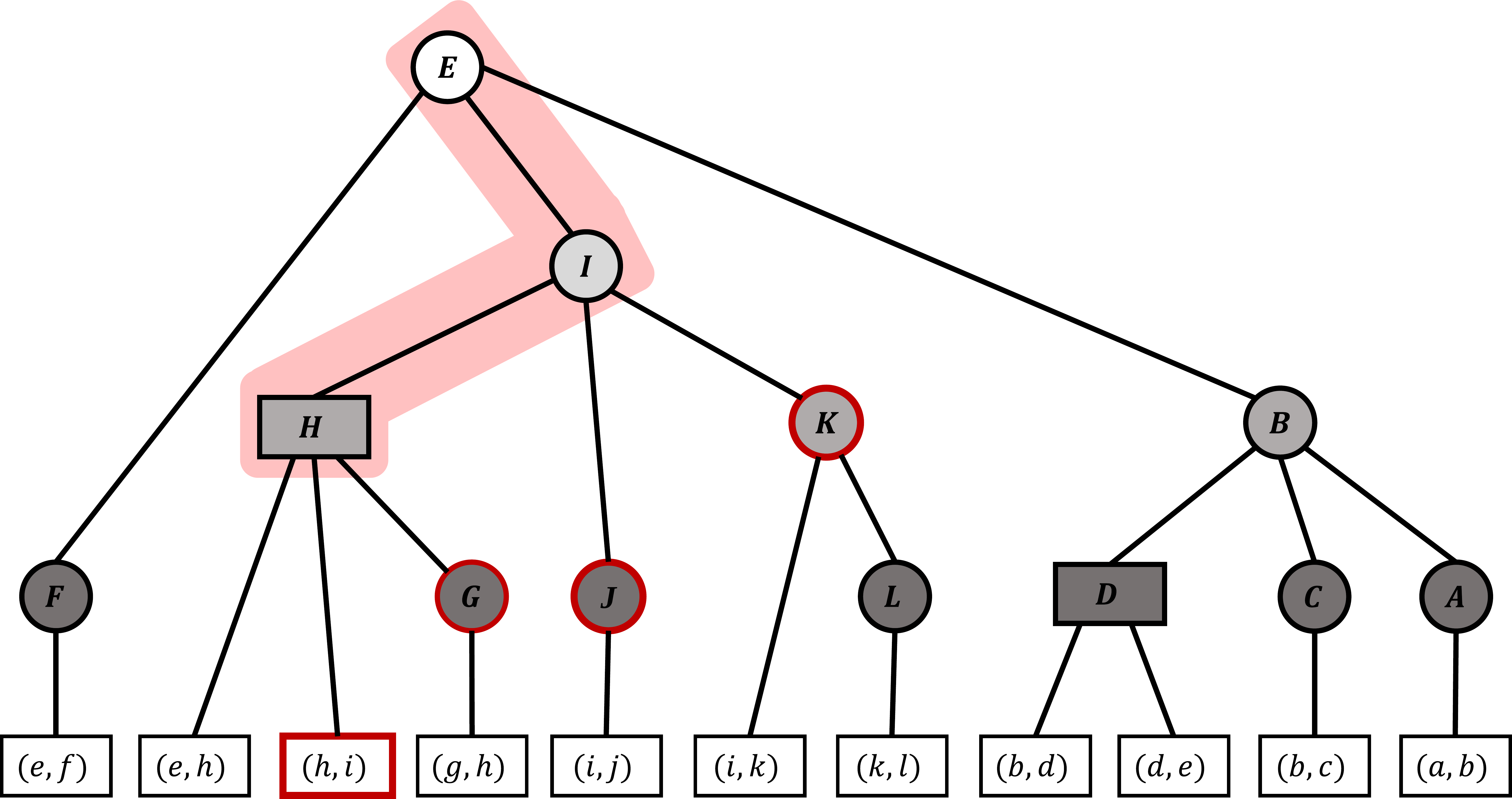}
		\caption{The clusters in the decomposition are children of the path from $H$ to its root cluster. The path in the RC-Tree is
			highlighted in red and the chosen clusters have a bold red outline.}
	\end{subfigure}
	\caption{An example of a subtree decomposition in an RC-Tree.}
	\label{fig:subtree-decomposition-example}
\end{figure}

An example of a subtree decomposition and the corresponding clusters in the RC-Tree is shown in Figure~\ref{fig:subtree-decomposition-example}. We
now prove a lemma that will come in handy during the proof. Consider some cluster $U$ represented by $u$ and one of its boundary vertices $b$. We define the
``subtree growing out of $b$'' with respect to $u$ to be the subtree rooted at $b$ oriented such that the tree was rooted at $u$.

\begin{lemma}\label{lem:boundary-subtree}
	Consider a cluster $U$ represented by $u$, and a boundary vertex $b$. The subtree growing out of $b$ with respect to $u$ as the tree
	root can be decomposed into a set of disjoint RC clusters which are all direct children of the path from $U$ to the root cluster.
\end{lemma}

\begin{proof}
	The situation is depicted in Figure~\ref{fig:subtree-remainder}. The goal is to collect a set of disjoint clusters that covers the subtree
	rooted at $b$, all of which are on the RC-Tree path from $U$ to the root cluster. To do so, observe that when $b$ contracts, it forms a
	larger cluster $B$. This cluster either contains $U$ as a direct child, or possible as a descendent further down the tree if the other
	boundary vertex of $U$ contracted earlier. $B$ is therefore an ancestor of $U$, and furthermore, all of the other children of $B$ must be
	contained inside the subtree of interest.
	
	We now consider the boundary vertices of $B$. If $B$ shares a boundary vertex with its child that contains $u$, we ignore that, since
	it is on the opposite side of $u$ to $b$ on the left side of $u$ in Figure~\ref{fig:subtree-remainder}) and hence is not contained in the subtree
	of interest. If $B$ has no other boundary vertex, then $B$ therefore contains the entirety of the subtree of interest and we are done.
	Otherwise, if $B$ has any other boundary vertex $b'$ (it may have up to two if $B$ is a binary cluster and the child containing $U$ is a
	unary cluster), it is on the opposite side of $b$ to $u$ (on the right side of $b$ in Figure~\ref{fig:subtree-remainder}),
	and hence is contained in the subtree of interest. We therefore simply recursively repeat this process with $b'$ (possibly two of them),
	identifying the ancestor $B'$ of $B$ where $b'$ contracts, and collecting all of its children except the one containing $b$. Once we run out of
	boundary vertices on the side of $b$, we have completed the subtree, and since every cluster collected was a child of an ancestor
	of $U$, it is a valid decomposition.
\end{proof}

\begin{figure}
	\centering
	\includegraphics[width=0.65\textwidth]{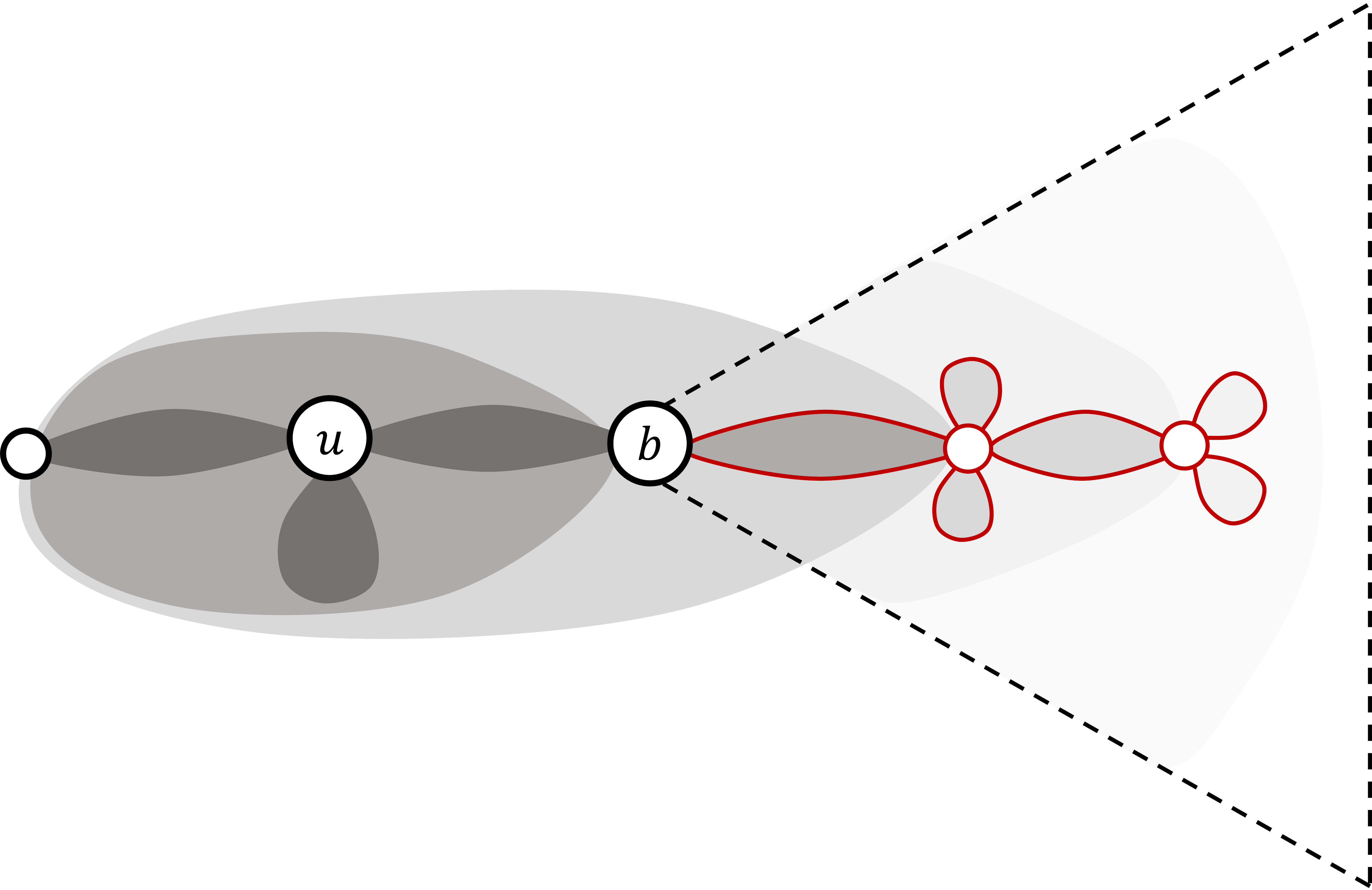}
	\caption{Constructing a decomposition of the subtree growing out of $b$ with respect to $u$ as the tree root. The process essentially walks along the boundary
		vertices to the right by following their contractions in the RC-Tree and collecting all of the adjacent clusters along the way.}\label{fig:subtree-remainder}
\end{figure}

\begin{proof}[Proofs of Theorem~\ref{thm:subtree-decomposition-property}]
	We prove Theorem~\ref{thm:subtree-decomposition-property} constructively, similar to Theorem~\ref{thm:path-decomposition-property}. The key idea is
	depicted in Figure~\ref{fig:subtree-decomposition-proof}. Consider the vertex $u$ and the cluster $U$ that it represents. $U$ consists of a constant
	number of child clusters, with at least one and at most two binary children, and some number of unary children. The subtree rooted at $u$ with respect
	to $p$ essentially consists of the entire tree except for anything in the direction of $p$. Furthermore, since $p$ is adjacent to $u$, we know that
	$p$ is either contained within one of the child clusters, or it is one of the boundary vertices. The second case happens if a binary child of $U$ is
	a single-edge base cluster with $p$ as the other endpoint.
	
	So, to construct a subtree decomposition, we start by taking all of the children of $U$, except for the one that contains/is adjacent to $p$.
	Then, for each binary child with boundary $b$, unless it is the one that contains/is adjacent to $p$, we add to the decomposition
	the contents of the subtree growing out of $b$. By Lemma~\ref{lem:boundary-subtree}, the subtree growing out of $b$ is decomposable into
	children of the path from $U$ to the root cluster, so this results in a valid subtree decomposition.
\end{proof}

\begin{figure}
	\centering
	\includegraphics[width=0.65\textwidth]{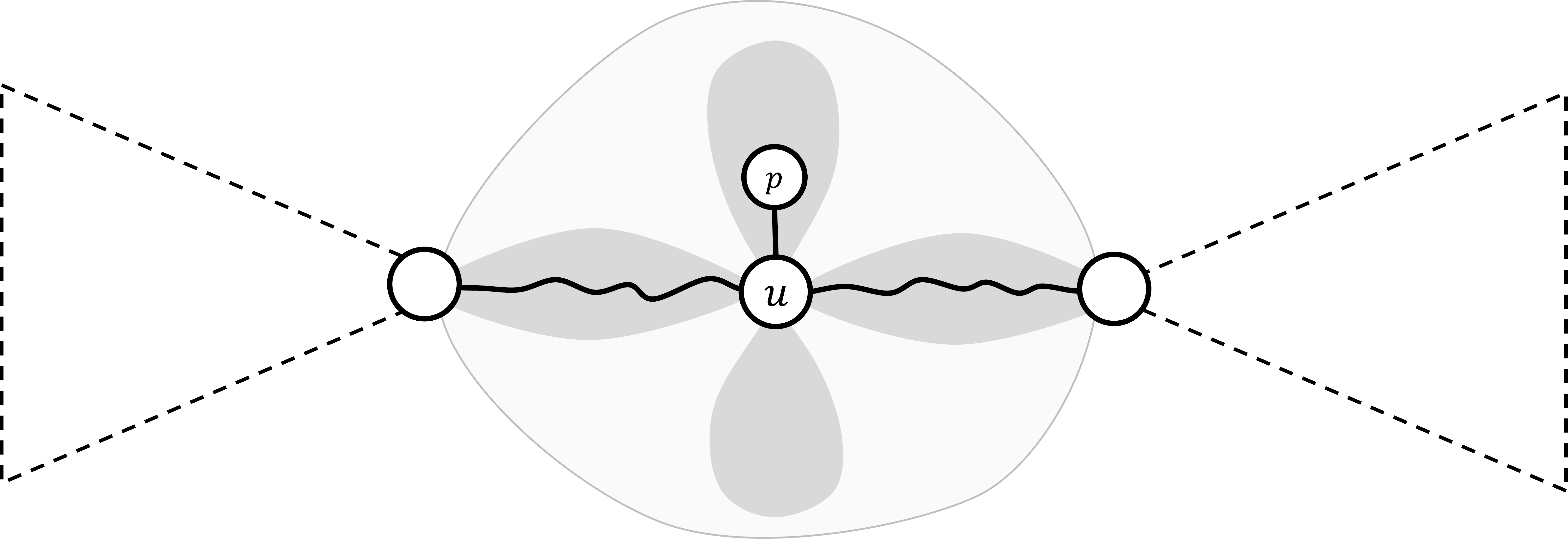}
	\caption{Constructing a subtree decomposition of the subtree rooted at $u$ with respect to $p$ as the parent. The subtree consists of (1) the adjacent clusters except the
		one containing/adjacent to $p$ and (2) the subtrees growing out of the boundary vertices unless adjacent to the cluster containing/adjacent to $p$.}\label{fig:subtree-decomposition-proof}
\end{figure}

\subsection{The Balance Property}

Given an input tree, there is no guarantee that there is a unique RC-Tree that encodes it. Indeed there are likely to be many. For example, one can always construct the ``rake tree'' of an input tree by performing only rake operations. The resulting RC-Tree unfortunately will be just as imbalanced as the input tree, so this is not very useful. The algorithms presented by Acar et al~\cite{acar2020batch} and Anderson and Blelloch~\cite{anderson2024} have the important property that at every round of contraction,
at least a constant fraction of the vertices are removed from the tree. Specifically, Acar et al.'s algorithm ensures that at each round, the number of vertices that survive is at most $7/8^\text{th}$s of the number of vertices in expectation. Similarly, the deterministic algorithm of Anderson and Blelloch ensures that at most $5/6^\text{th}$s of the vertices survive. We refer to this constant as $\beta$.

It is easy to show that an RC-Tree constructed with this property will have height $O(\log n)$, but we note that simply having $O(\log n)$ height alone is not sufficient for batch queries to be efficient. It will be typical of batch RC-Tree queries to start at a set of nodes related to the queries, then aggregate some information along all of the paths to the root from those nodes. The total work of such queries is typically proportional to the total number of unique nodes in this set of paths. The following theorem shows that on a tree constructed with this balance property, the number of such nodes is $O\left(k\log\left(1+\frac{n}{k}\right)\right)$.

\begin{theorem}[Balanced RC-Trees~\cite{acar2020batch}]\label{thm:rc-batch-theorem}
	Consider an RC-Tree constructed such that at most a constant $\beta$-fraction of the nodes survive tree contraction at each round. Given any $k$ nodes in such an RC-Tree, the total number of nodes in the union of the paths from those $k$ nodes to the root is at most $O\left(k\log\left(1+\frac{n}{k}\right)\right)$ (in expectation if $\beta$ is in expectation).
\end{theorem}

\subsection{Batch Connectivity Queries}

\begin{proof}[Proof of Theorem~\ref{thm:batch-connectivity-performance}]
	The batch find-representative algorithm performs work proportional to the number of unique nodes visited,
	which by Theorem~\ref{thm:rc-batch-theorem} is $O\left(k\log\left(1+\frac{n}{k}\right)\right)$ (in expectation if the RC-Tree is randomized) for $k$ queries. The span is
	the height of the RC-Tree, which is $O(\log n)$ for a balanced RC-Tree (w.h.p.\ if the RC-Tree is randomized). The algorithm
	answers a batch of find-representative queries, which is sufficient to answer a batch of connectivity queries by checking
	whether the representatives of each pair are the same. Therefore the final work to answer the connectivity queries is $O\left(k + k\log\left(1+\frac{n}{k}\right)\right)$.
\end{proof}

\subsection{Batch Subtree Queries}

Like connectivity, our algorithm for batch subtree queries is similar to the sequential single-query algorithm but with the added
care to avoid performing redundant work and redundant computations. Naively running $k$ subtree queries in parallel would result
in a lot of redundant work since many subtrees have substantial overlap, so our goal is to avoid that redundancy. Unlike connectivity,
it is much less obvious how to do so, since different subtree queries, although they may traverse the same root-to-leaf path in
the RC-Tree, may accumulate different chunks of information along the way. Batch connectivity queries on the other hand were simpler
since the answer for every RC node on the same path was guaranteed to be the same. Figuring out precisely how to break up the computation to
avoid redundancy will be the key insight of this algorithm.

The key idea is that a subtree rooted at $u$ can be decomposed into (1) all but one of the clusters adjacent to $u$ when $u$ contracts (i.e., all but one of the children of the cluster $U$ represented by $u$), and (2) one or two subtrees growing out of the boundary vertices of $U$ (See Figure~\ref{fig:subtree-decomposition-proof}). The contributions of Part (1) can be computed in constant time since those clusters are
children of $U$ in the RC-Tree. The complexity therefore lies in batching the computation of Part (2), the contributions of the subtrees growing
out of the boundary vertices.

\subsubsection{Algorithm: Batch subtrees}

The algorithm begins with a bottom-up computation that stores on each cluster the total aggregate weight of the contents of that
cluster. This is stored as an augmented value and hence already available at query time. The key step of the batched algorithm
is the subsequent top-down computation. Specifically, our algorithm computes for every relevant cluster, the contributions of the subtrees
growing out of its boundary vertices. The first step is the same as the batch find-representative algorithm. Starting from every
query vertex $u_i$, it walks up the tree concurrently and marks every ancestor of every $u_i$. This marks out a subtree of
the RC-Tree consisting of $O\left(k\log\left(1+\frac{n}{k}\right)\right)$ relevant clusters, specifically, every boundary vertex of every ancestor of every query vertex
$u_i$.

The algorithm computes these values starting at the top of the RC-Tree with a top-down computation. Given a node $U$
with representative $u$, we want to compute for each of its boundary vertices $b$, the total weight in the subtree growing out of $b$.
Consider the node $B$ represented by $b$. $B$ is an ancestor of $U$ in the RC-Tree and hence the total contributions of the subtrees
growing out of its boundaries have already been computed. $B$ has one child $U'$ which is either $U$ itself or an ancestor of $U$.
The algorithm collects the contributions of every other child of $B$, since these are contained in the subtree growing out of $b$.
Then, it considers each boundary of $B$, and for each boundary $b'$ of $B$ \emph{not shared with} $U'$, it recursively adds the
contribution of the subtree growing out of $b'$. Since the value of the subtree growing out of $b'$ was computed earlier,
looking this up takes constant time. Therefore, the contribution of the subtree growing out of $b$ with respect to $u$ can
be computed in constant time.

After completing this top-down computation, the algorithm can then answer every query $(u,p)$ in parallel by summing the contributions
of the children of $U$ except the one containing/adjacent to $p$ and then adding the contributions from the top-down computation
of the subtrees growing out of the boundary vertices of $U$ except for at most one that is adjacent to the cluster containing/adjacent to $p$.
Thus we have the following theorem.

\begin{proof}[Proof of Theorem~\ref{thm:batch-subtree-performance}]
	The preprocessing step visits every ancestor of every query root $u$, of which there are at most $O\left(k\log\left(1+\frac{n}{k}\right)\right)$ (in expectation
	if randomized) according to
	Theorem~\ref{thm:rc-batch-theorem}. The top-down computation then examines the children of each node, of which there are
	a constant number, and looks up the contributions of at most four previously computed subtrees (there are at most two boundary
	vertices of the current node, and for each they consider the at most two boundary vertices of the corresponding ancestor node).
	Therefore a constant amount of work is required per node, resulting in $O\left(k\log\left(1+\frac{n}{k}\right)\right)$ preprocessing work (in expectation if randomized), and by parallelising
	the top-down traversal, a span of $O(\log n)$ (w.h.p.\ if randomized). After preprocessing, each query $(u,p)$ is then answered in constant
	time by examining the constant number of children of $U$ and looking up at most two precomputed contributions for the boundary
	vertices. Since there are at most $O(n)$ possible subtrees, the final work bound is $O\left(k + k\log\left(1+\frac{n}{k}\right)\right) = O\left(k\log\left(1+\frac{n}{k}\right)\right)$.
\end{proof}

\subsection{Batch Path Queries with Inverses}

We reduce a batch of path queries over a commutative group to an invocation of our batch-LCA algorithm from Section~\ref{sec:batch-lca} and a simpler
batch query. For each path query $(u,v)$, we compute the LCA of $u$ and $v$ with respect to the RC-Tree root $r$. Then, we solve a batch of path queries
of the form $\textproc{PathSum}(r,x)$.  The solution to the query $(u,v)$ is then calculated as $\textproc{PathSum}(r,u) + \textproc{PathSum}(r,v) - 2\textproc{PathSum}(r,\text{LCA}(u,v))$.

The algorithm maintains augmented values on each binary cluster corresponding to the sum of the weights on the cluster path. Since this can be computed
by summing the weights of the binary children, this is a simple bottom-up computation that uses constant time and space per cluster. No augmented values
are stored on the unary clusters. At query time, the algorithm begins with the usual process of marking all of the ancestor clusters of every vertex
$u$ and $v$ considered in the query.  This marks out a set of $O\left(k\log\left(1+\frac{n}{k}\right)\right)$ nodes in the RC-Tree (Theorem~\ref{thm:rc-batch-theorem}). It then
performs a top-down computation on those nodes which computes the total weight on the path from $r$ to the representative vertex of the cluster.

This top-down computation can be computed in constant time per cluster by considering the values of the boundary vertices. If the current cluster is
a unary cluster, then the total weight from the root is just the total weight to the boundary vertex plus the total weight of the cluster path of
the binary child. If the current cluster is a binary cluster, then the total weight to the root is the total weight to the boundary vertex that
is closer to the root, plus the weight of the cluster path of the adjacent binary child.  Determining which boundary vertex is closer to the root
can be done using the same top-down computation as the batch-LCA algorithm. Once each marked cluster has the total weight from the root to its
representative, the input queries can be solved.

\begin{proof}[Proof of Theorem~\ref{thm:batch-path-query-group-performance}]
	The augmented values need only examine the values of child clusters, so they take constant time.
	The computation of the batch-LCAs takes $O\left(k + k\log\left(1+\frac{n}{k}\right)\right)$ work (in expectation if randomized) and $O(\log n)$ span (w.h.p.\ if randomized)
	by Theorem~\ref{thm:batch-lca}. Marking all the ancestors of the input vertices also takes $O\left(k\log\left(1+\frac{n}{k}\right)\right)$ work and $O(\log n)$ span.
	The top-down computation looks at the at most two boundary vertices of the cluster and one of the children of the cluster and hence a
	constant amount of work is required per node. This results in $O\left(k + k\log\left(1+\frac{n}{k}\right)\right)$ preprocessing work (in expectation if randomized), and by parallelising
	the top-down traversal, a span of $O(\log n)$ (w.h.p.\ if randomized). After preprocessing, each query $(u,v)$ is answered in constant time by looking
	up the values on the clusters represented by $u$, $v$, and LCA$(u,v)$, which takes constant time.
\end{proof}

\subsection{Batch Path Minima Queries}

Anderson et al.~\cite{anderson2020work} give an algorithm that, given an RC-Tree for an edge-weighted tree on $n$ vertices and a set
of $k$ marked vertices, they can produce a \emph{compressed path tree} with respect to the $k$ marked vertices in $O\left(k\log\left(1+\frac{n}{k}\right)\right)$ work and $O(\log n)$ span.
Our algorithm for batch path minimas first produces the compressed path tree with respect to the union of the endpoints of the query vertices $u_i$ and $v_i$.
By design of the compressed path tree, the answers to the path queries are the same if they are evaluated on the compressed path tree
as if they were evaluated on the original tree.  At this point, since the compressed path tree has just $O(k)$ vertices, we can run a static offline algorithm that
evaluates the queries and no longer need the RC-Tree. Our algorithm uses the subroutine of King et al.'s parallel MST verification algorithm~\cite{king1997optimal}
which evaluates a static offline batch of path minima queries in $O(n + k)$ work and $O(\log n)$ span, where $n$ is the size of the tree and $k$ is the
number of query edges. The results of these queries is the solution. Given the algorithm for path minima queries, the algorithm can handle path maxima queries
by reversing the result of each comparison, or equivalently, negating the cost of each edge.

\begin{proof}[Proof of Theorem~\ref{thm:batch-path-minima-performance}]
	Building the compressed path tree takes $O\left(k\log\left(1+\frac{n}{k}\right)\right)$ work (in expectation if randomized) and $O(\log n)$ span (w.h.p.\ if randomized)~\cite{anderson2020work}. The static offline algorithm of King et al.~\cite{king1997optimal} runs in linear work and logarithmic
	span in the size of the tree and number of queries.  Since the compressed path tree has $O(k)$ vertices and we evaluate $k$ queries,
	it takes $O(k)$ work and $O(\log k)$ span. Therefore the total work is $O\left(k + k\log\left(1+\frac{n}{k}\right)\right)$ (in expectation if randomized) with $O(\log n)$ span (w.h.p.\ if randomized).
\end{proof}

\subsubsection{Batch Nearest-Marked-Vertex Queries}

Our algorithm for nearest marked vertices maintains several augmented values via a bottom-up computation, similarly to the sequential single-query algorithm~\cite{acar2005experimental}. For each cluster,
it maintains (1) the nearest marked vertex in the cluster to the representative, (2) the nearest marked vertex in the cluster to each boundary vertex, and (3), if a binary
cluster, the length (total weight) of the cluster path. If a cluster contains no marked vertices, it stores a value of $\mathbf{null} / \infty$ for (1) and (2). Each of
these can be maintained in constant time given the values of the children.

\begin{enumerate}[leftmargin=*, label=(\arabic*)]
	\item If the representative vertex is marked, then it is the nearest marked vertex to itself. Otherwise, since the representative is the boundary vertex
	shared by the children, the nearest marked vertex to the representative in the cluster is just the nearest out of all the nearest marked vertices to that
	boundary vertex in the children.
	\item Each boundary vertex is shared with one of the children. The nearest marked vertex to it in the cluster is either the nearest marked vertex to it in
	that child cluster or it is in one of the other children and hence is the same as the nearest marked vertex to the representative.
	The nearest of the two can be compared by adding the weight of the cluster path to the distance to the nearest to the representative.
	\item The total weight on the cluster path is just the sum of the weights of the two binary children's cluster paths, or the edge weight if the cluster is
	a base cluster.
\end{enumerate}

\noindent To perform updates such as \tcbatchmark{} and \tcbatchunmark{}, the algorithm just sets the marks on the corresponding vertices and then propagates the augmented
values up the RC-Tree to every ancestor cluster containing those vertices. This takes $O\left(k\log\left(1+\frac{n}{k}\right)\right)$ work and $O(\log n)$ span. The time to propagate an edge-weight
update would be the same.

The augmented values store what we will call the \emph{locally} nearest marked vertices, i.e., the nearest marked vertices inside the same cluster. At query time, the
algorithm performs a top-down computation to compute the \emph{globally} nearest marked vertices, i.e., the actual nearest marked vertices in the entire tree. Specifically,
after performing the standard preprocessing step of marking every ancestor cluster of the input vertices, it will perform a top-down computation that will compute for each
marked cluster, the (globally) nearest marked vertex to the representative vertex.

For a root cluster, since it contains its entire subtree, the (globally) nearest marked vertex to the representative is the same as the augmented value. Otherwise, given
a non-root cluster there are two cases. Either the nearest marked vertex is inside the cluster or it is outside the cluster. The nearest marked vertex inside the cluster
is read from the augmented value. If the nearest marked vertex is outside the cluster, then the path from the representative to the nearest marked vertex must pass through
one of the boundary vertices, and hence the nearest marked vertex is just the closer of the nearest marked vertex to each of the boundary vertices, which have already
been computed earlier in the top-down computation. The algorithm therefore computes the distances from the representative vertex to the nearest marked vertices of the
boundaries by adding the weights of the cluster paths that connect them, then takes the closest of the options.

After the top-down computation, the answers can be read off the values of the represented clusters of the query vertices in constant time. We therefore have the following result.

\begin{proof}[Proof of Theorem~\ref{thm:batch-marked-performance}]
	The augmented values need only examine the values of child clusters, so they take constant time. Marking all the ancestors of the input vertices
	also takes $O\left(k\log\left(1+\frac{n}{k}\right)\right)$ work and $O(\log n)$ span (Theorem~\ref{thm:rc-batch-theorem}). The top-down computation looks at the augmented values
	and at the augmented values of the at most two boundary vertices of the cluster and hence a constant amount of work is required per node. This results
	in $O\left(k\log\left(1+\frac{n}{k}\right)\right)$ preprocessing work (in expectation if randomized), and by parallelising the top-down traversal, a span of $O(\log n)$ (w.h.p.\ if randomized).
	After preprocessing, each query $v_i$ is answered in constant time by looking up the value on the cluster represented by $v_i$.
\end{proof}

\subsection{Details on LCA}
\subsubsection{Introduction and Definitions}

LCA is a generally useful operation on a tree. LCA generalizes connectivity and isAncestor queries. In an RC tree, the common boundary is an LCA. LCA has also been used for maximum weight matchings \cite{Ga90}. In this work, we provide the first implementation of LCA in a parallel batch dynamic tree: previous parallel batch dynamic tree structures, like Euler-Tour trees, were not capable of handling LCA queries \cite{tseng2019batch}. Our main result in this section is a proof of Theorem \ref{thm:batch-lca}. 

Before proving Theorem \ref{thm:batch-lca} we highlight some folklore results about LCA. 

\begin{definition} Let $T$ be an unrooted tree and $a,b,r$ be vertices in $T$. We say that $a$ is an {\bf ancestor} of $b$ (with respect to root $r$), denote $isAncestor_T(a,b,r)$, if the (unique) path from $b$ to $r$ includes $a$. Note that we consider a vertex to be its own ancestor. We say $a$ is a proper ancestor of $b$ if $isAncestor(a,b)$ and $a \ne b$. \end{definition}

In proofs, it is often convenient to define LCA in terms of minimizing tree distance. Write $path(a,b,c)$ if $a$ is on the path from $b$ to $c$.

\begin{definition} \label{distDef} Define $d(u,v)$ as the unweighted path distance from $u$ to $v$ and define $D_{u,v,r}(c):=d(u,c) + d(v,c) + d(r,c)$. It is known we could also define $LCA(u,v,r')$ as the distance minimum of $(u,v,r)$, that is, $argmin_{c \in V} D_{u,v,r}(c)$. For the forest extension of LCA, we would define $d(u,v)=\infty$ if $u$ and $v$ are not connected. \end{definition}

\begin{lemma} \label{d3} Suppose that $c=argmin_{x \in V} D_{u,v,r}(x)$. Then $c$ is on the path from (WLOG) $u$ to $r$. \end{lemma}

\begin{proof} WLOG consider the path from $u$ to $r$. Let $j$ be the closest node to $c$ on this path. Since a tree contains no cycles, $d(u,c)=d(u,j)+d(j,c)$ and $d(r,c)=d(r,j)+d(j,c)$. By the triangle inequality, $d(v,c) \ge d(v,j)-d(j,c)$. Then $D_{u,v,r}(c)\ge D_{u,v,r}(j)+d(j,c)$. Since $c$ minimizes $D_{u,v,r}$, $D_{u,v,r}(c) \le D_{u,v,r}(j)$. Thus, $d(j,c)=0$. Thus $c$ is on the path from $u$ to $r$. \end{proof}

\begin{lemma} There is a unique minimizer $c$ of $D_{u,v,r}$. \end{lemma}

\begin{proof} Suppose BWOC $c' \ne c$ also minimizes $D_{u,v,r}$. By Lemma \ref{d3}, $path(c,v,r)$ and $path(c',v,r)$. Since there is a unique node $c$ distance $d(v,c)$ away from $v$ on the path from $v$ to $r$, $d(v,c) \ne d(v,c')$. 

Similarly, $d(u,c) \ne d(u,c')$. WLOG, $d(u,c) < d(u,c')$. Thus $d(u,c')=d(u,c)+d(c,c')$ and $d(c,r)=d(c,c')=d(c',r)$. Thus $d(c,v)=D(c)-d(c,u)-d(c,r)=D(c')-d(c,u)-d(c,r)=d(c',u)+d(c',v)+d(c',r)-d(c,u)-(d(c,c')+d(c',r))=d(c',v)$, contradicting that $d(v,c) \ne d(v,c')$. Thus $c=c'$.   \end{proof}

\begin{lemma} Definitions \ref{ancDef} and \ref{distDef} are equivalent. \end{lemma}

\begin{proof} 
Let $c=argmin_{x \in V} D_{u,v,r}(x)$. Note that $c$ exists and is unique, and note that the LCA also exists. WLOG consider the path from $u$ to $r$. By Lemma \ref{d3}, $path(c,u,r)$, so $isAncestor(c,u,r)$. By the same argument, note that $isAncestor(c,v,r)$.

Let $c'$ be an ancestor of $u$ and $v$. Suppose BWOC that $c$ is a proper ancestor of $c'$. Then $d(c',u)=d(u,c)-d(c,c')$, $d(c',v)=d(v,c)-d(c,c')$, and $d(c',r)=d(c,c')+d(c,r)$, so $D(c',u,v,r)=D_{u,v,r}(c)-d(c,c')<D_{u,v,r}(c)$, contradicting that $c$ minimizes $D$. Thus $c'$ is an ancestor of $c$. Thus $c$ is $LCA(u,v,r)$. 
\end{proof}

From this definition, we see that $u,v,r$ are symmetric parameters to an LCA query. We also can argue that if 3 vertices are on a path, the middle vertex is the LCA; this in spirit is a converse of Lemma \ref{d3}.

\begin{corollary} The arguments, $u,v,r$, to LCA are symmetric; that is, $LCA(u,v,r)=LCA(v,u,r)=LCA(r,u,v)$.  \end{corollary} 

\begin{lemma} \label{path3} Let $T$ be an unrooted tree and $a,b,c \in T$. If $path(a,b,c)$ then $a=LCA(a,b,c)$. \end{lemma}

\begin{proof} Let $x \in T$. Observe that $D(a,a,b,c)=d(a,a)+d(a,b)+d(a,c)=d(b,c) \le d(x,b)+d(x,c) \le d(x,b)+d(x,c)+d(x,a)=D(x,a,b,c)$. Thus by Definition \ref{distDef}, $a$ is the LCA. \end{proof}

\subsubsection{Batch dynamic LCA}

For parallel batch dynamic LCA, we will need the following lemmas. 

\begin{lemma} \label{logic} (Folklore) $LCA(u,v,r')=LCA(u,v,r) \oplus LCA(u,r',r) \oplus LCA(v,r',r)$ \end{lemma}

The proof of \ref{logic} follows from casework on the relative positions of $u,v,r,r'$ and is thus omitted. Note that $|\{LCA(u,v,r),LCA(u,r',r),LCA(v,r',r)\}| \le 2$, lest there be a cycle, so $LCA(u,v,r) \oplus LCA(u,r',r) \oplus LCA(v,r',r)$ is well-defined. 

\begin{lemma} \label{sLCA} One can preprocess a static LCA structure on a tree with $n$ nodes in $O(n)$ work and polylog depth, such that queries take constant time \cite{SV88}. \end{lemma} 

\begin{lemma} \label{la} One can preprocess a level ancestors structure on a tree with $n$ nodes in $O(n)$ work and polylog depth such that queries take constant time \cite{BV94}. \end{lemma} 

\begin{lemma} \label{highAncestor} Given a cluster $B$ and vertex $u \in B$ in an RC tree rooted at $r$, after preprocessing we can find the highest ancestor of $u$ in $B$ in constant time. \end{lemma}
\begin{proof} For each RC-tree child $X$ of $B$, we choose $X$ if $isAncestor(X,U,R)$. There are a constant number of children of $B$, and checking isAncestor takes constant time with the static LCA data structure, for constant time overall.  \end{proof}

\begin{lemma} \label{closePath} Given a binary cluster $B$ and vertex $u \in B$, after preprocessing we can find the closest vertex to $u$ on the cluster path of $B$ in constant time.  \end{lemma} 
\begin{proof} 

Using the ancestor bitset, we examine $u$'s ancestors in $B$. If $u$ is in the cluster path, return $u$. Otherwise, use the bitset to find the level highest unary cluster $W$ in $B$ containing $u$. We extract $W$ via the level ancestors structure, which by Lemma \ref{la} takes constant time, then return $W$'s boundary vertex. \end{proof}




The procedure batchLCA takes a list of queries and an unrooted forest, and returns the LCA of each query. From the queries, we do a bottom-up sweep in the RC tree to determine the nodes involved in the calculation.Sweeping top down, we separate the queries by their tree within the forest. For queries on elements belonging to multiple trees, we mark that the query nodes are not connected (so the LCA does not exist). Then, in parallel we handle each tree component. Consider a tree component with RC root $r$ and query $LCA(u,v,r')$. We calculate $LCA(u,v,r), LCA(u,r',r),$ and $LCA(v,r',r)$ using the procedure batchFixedLCA, then apply Lemma \ref{logic} to find $LCA(u,v,r')$.

The procedure batchFixedLCA takes a tree rooted at $r$ and a list of queries (all queries with root $r$) and returns the LCA of each query. First, we mark the tree nodes involved in the computation (bottom-up sweep). On this subtree, we preprocess a static LCA structure (Lemma \ref{sLCA}) and a level ancestors structure (Lemma \ref{la}). For every cluster involved, we calculate the closest boundary vertex to the root. Each cluster stores a bit for each of its ancestors denoting whether it is unary or binary. 

Then, using the static LCA structure, we find the common boundary $c$ of each query pair $(u,v)$. If $r$ is $u,v,$ or $c$, then the LCA is $r$. If $c$ is $u$ or $v$ (WLOG $v$), then let $X$ be the ancestor of $U$ that is a direct child of $V$. If $X$ is unary, then $c$ is the LCA. If $X$ is binary, we case on the position of $r$. If $r$ and $x$ are on opposite sides of $v$, then $c$ is the LCA. Otherwise, $r$ and $x$ are on the same side of $v$, so the vertex on the cluster path of $X$ closest to $u$ is the LCA. 

When $c \not\in \{u,v\}$, look at the clusters $X$ and $Y$ that are ancestors of $U$ and $V$ (respectively) and are direct children of $C$. If $c$ is in between $x$ and $r$ and in between $y$ and $r$, then $c$ is the LCA. Otherwise, (WLOG) $c$ is in between $x$ and $r$ but not $y$ and $r$. Then the vertex on the cluster path of $Y$ closest to $v$ is the LCA. 

Now we will prove Theorem \ref{thm:batch-lca}.

\begin{proof}
Proof of work/depth: Let $L=k \log(1+n/k)$. By Theorem \ref{thm:rc-batch-theorem}, we operate on a $O(L)$ subset of the RC tree. Sweeps each take $O(L)$ work and $O(\log n)$ depth. By Lemmas \ref{sLCA} and \ref{la} static LCA and level ancestors are linear work in the tree size $L$. Thus, all of the preprocessing fits within our bounds.

Within the casework for each query, finding the highest ancestor within a cluster takes constant time by Lemma \ref{highAncestor}, and finding the closest vertex on a cluster path takes constant time by Lemma \ref{closePath}.   

If $r$ is (WLOG) $u$, then the LCA is $r$ because $LCA(u,v,r)=LCA(u,r,v)=LCA(u,u,v)=u$. If $r=c$, then since $r$ is the RC LCA, by Lemma \ref{d3} we have $path_{RC}(r,u,v)$. Note that the RC tree root partitions the original tree because it clusters last. Thus, it follows that $path(r,u,v)$. Thus $r$ minimizes $D_{r,u,v}$, so $r$ is the LCA.  

Now suppose that $c$ is $u$ or $v$ (WLOG $v$). By definition of LCA, $isAncestor(v,u,r)$. Because $v$ is an ancestor of $u$ in the RC tree, $v$ is a boundary vertex of an ancestor $X$ of $U$ in the RC tree. Note that $r$ is not inside $X$ because $r$ contracts last. If $X$ is unary, then $r$ and $x$ are on opposite sides of $c$, so $path(c,r,x)$, and so $c$ is the LCA by Lemma \ref{path3}. If $X$ is binary, we case on the position of $r$. If $r$ and $x$ are on opposite sides of $c$, then by Lemma \ref{path3}, $c$ is the LCA. Otherwise, $r$ and $x$ are on the same side of $v$, so the vertex on the cluster path of $X$ closest to $u$ is the LCA, which we find with Lemma \ref{closePath}. 

When $c \not\in \{u,v\}$, look at the clusters $X$ and $Y$ that are ancestors of $U$ and $V$ (respectively) and are direct children of $C$. If $c$ is in between $x$ and $r$ and in between $y$ and $r$, then $c$ is the LCA. Otherwise, (WLOG) $c$ is in between $x$ and $r$ but not $y$ and $r$. Then the vertex on the cluster path of $Y$ closest to $v$ is the LCA. \end{proof}

\subsubsection{LCA and ternarization} We give the proof of Lemma \ref{zeroEdge2}. 
\begin{proof}Suppose BWOC that $\lnot path(c,c',u)$ and $\lnot path(c,c',v)$. Then there exists a path from $u$ to $v$ without $c$, contradicting Lemma \ref{d3}. Thus $path(c,c',u)$ or $path(c,c',v)$. Similarly applying to $v$ and $r$, and $u$ and $r$, we get that ($path(c,c',u)$ or $path(c,c',r)$) and ($path(c,c',v)$ or $path(c,c',r)$). Thus, $c$ is present in at least 2 of 3 paths $c'$ to $u$, $c'$ to $v$, and $c'$ to $r$. WLOG suppose that $path(c,c',u)$ and $path(c,c',v)$. Then by the triangle inequality $W_{u,v,r}(c') \le W_{u,v,r}(c)=w(u,c)+w(v,c)+w(r,c)=(w(u,c')-w(c,c'))+(w(v,c')-w(c,c'))+w(r,c) \le (w(u,c')-w(c,c'))+(w(v,c')-w(c,c'))+(w(r,c') + w(c,c')) = W_{u,v,r}(c') - w(c,c')$. Thus $w(c,c') \le 0 \implies w(c,c')=0$. \end{proof}

Remark: Let $T$ be a (possibly) high degree tree and $T'$ its ternarized form. Note that $w$ from Lemma \ref{zeroEdge2} is not a distance metric on $T'$ because $w(a,b)=0 \nRightarrow a=b$; however, $w$ does satisfy symmetry and the triangle inequality, and $w$ is a distance metric on $T$.

\newpage

\end{document}